\documentclass[12pt,oneside]{article}

\usepackage{amsthm}
\usepackage[charter]{mathdesign}  
\SetMathAlphabet{\mathcal}{normal}{OMS}{cmsy}{m}{n}
\SetMathAlphabet{\mathcal}{bold}{OMS}{cmsy}{b}{n}

\usepackage[T1]{fontenc}
\usepackage{parallel}
\usepackage{setspace}
\usepackage[round, authoryear]{natbib}
\usepackage{tikz}
\usepackage{caption}
\usepackage{graphicx}

\usepackage[toc]{appendix}
\usepackage{amsmath}
\usepackage{textcomp}
\usepackage{xcolor}
\definecolor{navyblue}{rgb}{0.0, 0.0, 0.5}
\definecolor{myred}{rgb}{0.39, 0.015, 0.15}
\definecolor{red(ncs)}{rgb}{0.77, 0.01, 0.2}
\usepackage[backref=page]{hyperref}
\hypersetup{
	colorlinks = true,
	allcolors = myred
}
\usepackage{cleveref}
\usepackage{thmtools}
\usepackage{thm-restate}

\author{
	Ester Sudano\thanks{School of Economics and Finance,
		Queen Mary University of London, UK. e-mail: e.sudano@qmul.ac.uk}
}

\title{Categorize and randomize:\\
		  a permissive model of stochastic choice}

\theoremstyle{plain}                    

\newtheorem{corollary}{Corollary}       
\newtheorem{lemma}{Lemma} 

\theoremstyle{definition}               
\newtheorem{definition}{Definition}[]
\newtheorem{example}{Example}      
\newtheorem{counterexample}{Counterexample}

\theoremstyle{remark}                   
\newtheorem{remark}{Remark}          

\newtheoremstyle{blueconj}{}{}{}{}{\bfseries\color{blue}}{.}{.5em}{}
\theoremstyle{blueconj}

\newcommand{\cind}{\hyperref[cind]{\textsc{c-Independence}}}
\newcommand{\cneu}{\hyperref[cneu]{\textsc{c-Neutrality}}}
\newcommand{\IIA}{\textsf{IIA}}
\newcommand{\nsc}{\textsf{NSC}}
\newcommand{\rum}{\textsf{RUM}}
\newcommand{\eba}{\textsf{EBA}}
\newcommand{\scc}{\textsf{SCC}}
\newcommand{\scwc}{\textsf{SCWC}}
\newcommand{\X}{\mathcal{X}}
\newcommand{\I}{\mathcal{I}}

\begin{document}
	
	\maketitle

	\begin{abstract}
		\noindent We model stochastic choices with categorization.
		The agent preliminarily groups alternatives in homogenous disjoint classes, then randomly chooses one class and randomly picks an item within the selected class.
		We provide a characterization for the model and, further, we show that stochastic choices with categorization are generated by a population of deterministic choice functions that share a consistent two-step structure.
		The model allows to describe the observed stochastic choice as the composition of independent subchoices, 
		and this composition preserves rationalizability by Random Utility Maximization.
		A generalization of the model subsumes Luce model and Nested Logit. 
		The characterizing properties disclose the behavioral foundations of the models and shed some light on the implications of Luce's \IIA{} and its weaker forms for choice behavior.  
	\end{abstract}
	


\section*{Introduction} \label{SEC:Introduction}

Categorization is a fundamental cognitive practice: 
reducing objects and events to a unifying abstract representation often helps in grasping vast information. 
Similarly, economic agents may simplify their decision process by first categorizing items, then choosing one within the selected class. 
The goal of this work is to formalize and analyze the behavior implied by such procedure.  
The resulting insights can prove useful for establishing relations of substitutability among options. 

We present \emph{stochastic choices with categorization} (\scc), a model of probabilistic choice informed by a consistent classification of options at disposal. 
In the model, the finite grand set $X$ of alternatives is partitioned in disjoint \textit{classes}. 
A stochastic choice with categorization (\scc) is the result of two steps: first, classes are considered and one is chosen, then items within it are compared and one is selected. 
Hence any \scc{} can be decomposed into \textquoteleft smaller\textquoteright\ independent choices, one defined on the set of classes together with a family of choices, each defined on a class. 
Each sub-choice can then be analyzed separately, so that an analyst can select the fittest model for each, to possibly enhance modelling performance and choice predictions. 
We refrain from further disciplining choice behavior, so as to allow for any form of behavioral decision-making. 

To illustrate, consider a decision maker who purchases a camera. 
She hierarchically evaluates a set of attributes she deems relevant for the choice. 
First, she focuses on a subset of macroscopic attributes, such as price range, brand and compact or mirrorless system, and chooses the class that satisfies the desired attributes, independently of available options within that class. 
Within the selected class, the final choice is made by evaluating particular attributes, such as sensor size and video performance, only among surviving items.
Similar procedures are ubiquitous in consumption choices, and are likewise employed in hiring decisions or by financial market investors who pick a portfolio within a suitable class of risk. 

The characterization of the model relies on the definition of a \emph{category}, interpreted here as a set of homogeneous options,
based on two behavioral assumptions. 
First, external items are irrelevant for determining the choice among items in a category.
For instance, if the decision maker classifies cameras as mirrorless or compact, the availability of a compact $c$ does not change her preference over mirrorless $a$ and $b$, as $c$ is too dissimilar to convey any extra information on $a$ and $b$. 
This requires a form of independence in choice behavior, whose scope is dictated by the categorization structure (\cind).
Second, any set of items in the same category equally affects the probability of choice for external options (\cneu), regardless of its cardinality and identity. 
This form of substitutability entails that, as long as some mirrorless camera is already available, the presence of an additional one does not alter the probability a compact camera is chosen from a menu; it can only vary the chances another mirrorless camera is selected. 

Typically, the variability in choice observations in experimental and marketing settings has been linked to variability in taste or noise in utility evaluations for an individual, or obtained as the result of the aggregation of choices in a population \citep{Thurstone1927, Luce1959, BlockMarshack1960}. 
In line with this, we clarify the relation to deterministic choice behavior 
by exhibiting a population of deterministic choice functions that generates stochastic choices with categorization. 
In particular, we show that any \scc{} can be represented by a probability distribution on a set of deterministic choices which share a two-step structure with respect to the same classification of items. 
We then introduce a generalization of the model, \textit{stochastic choice with weak categorization}, which allows choice probabilities for classes to be menu dependent. 
In its characterization, we employ a notion of \textit{weak category}, which only requires \cind.  
We explore the relation of both models to the class of Random Utility Models, and show that any \rum{} stochastic choice with categorization can be decomposed in \rum{} choices, and vice versa. 
Finally, stochastic choice with weak categorization subsumes Luce Model, Nested Stochastic Choice and Nested Logit. 
By contrasting these models to our definitions, we clarify the behavioral implications of Luce's \textsf{IIA} and weaker related axioms in literature.

\paragraph{Relationship with the literature.}

\cite{CarpentiereDoignon2025} generalizes the definition and characterization of \scc{} to the broader framework of stochastic choice correspondences, allowing the choice of a subset -- rather than a single item -- from a menu. 
While their work is primarily concerned with exploring properties that are preserved in the composition and decomposition of stochastic choices, with particular focus on rationalizability by Random Utility, our simpler setting allows to underscore the behavioral implications of the model and its relations to deterministic choice behavior. 
In this regard, here we briefly anticipate the connection to Nested Stochastic Choices (\nsc) \citep{Kovach2022}, a generalization of Nested Logit \citep{McFadden1978}, to which our models closely relate. 
In Nested Stochastic Choices, disjoint nests are specified by peculiar violations of Luce's Independence of Irrelevant Alternatives (\IIA) \citep{Luce1959} due to the \textit{similarity effect}, provided that
\IIA{} holds within nests. 
Denote $p(x, A)$ and $p(y,A)$ the probability that item $x$ and $y$ are chosen from menu $A$. 
Luce's \IIA{} requires the ratio $\frac{p(x, A)}{p(y, A)}$ to be fixed for any set $A$ in which $x$ and $y$ are available. 
To illustrate the similarity effect, we follow the famous example by \cite{Debreu1960} and \cite{McFadden1978}, and consider a decision maker choosing a way to commute. 
When the only options are a train ($t$) and a blue bus ($b$), choice probabilities satisfy $\frac{p(t, \{t, b\})}{p(b, \{t, b\})} = 1$. 
If a red bus ($r$) covering an identical route is made available, the agent regards it as a duplicate of $b$, and the probability of choosing $b$ is hurt. 
This is the similarity effect, resulting in the inequality $\frac{p(t, \{t, b, r\})}{p(b, \{t, b, r\})} > \frac{p(t, \{t, b\})}{p(b, \{t, b\})}$. 
In Nested Stochastic Choice, and Nested Logit, this is the only source of violation of Luce's \IIA.

Yet, evidence suggests that especially when items are comparable, but not exact replicas, choices can exhibit further context effects. 
Suppose item $a$ has better quality but worse economy than $b$. 
The \textit{attraction effect} occurs when the addition of $d$, inferior to $a$ both in quality and price, makes the choice of $a$ more likely \citep{HuberPaynePuto1982}.  
On the other hand, the introduction of an option $c$, with extremely good quality but poor economy, can make $a$ a compromise between $c$ and $b$ and enhance its chances of choice. 
This is known as the \textit{compromise effect} \citep{Simonson1989}. 
We depart from \nsc{} by allowing these common phenomena, and admitting failures of \IIA{} within classes. 
This relaxation of \IIA{} is precisely captured by \cind.
Crucially, these common context effects are inconsistent with the axioms of rationality in stochastic choice, and thus imply the lack of a random utility representation.
Our definitions of a (weak) category therefore captures relations of similarity, without requiring the existence of an underlying random utility function. 
We refer to Section~\ref{SEC:Relation to literature scc} for a thorough analysis of the relation with literature. 

\smallskip

The paper is outlined as follows. 
Section~\ref{SEC:Model} introduces the model, then characterized in Section~\ref{SEC:Characterization}.
In Section~\ref{SEC:deterministic resolvability}, we study stochastic choices with categorization as generated by a population of deterministic choices.
We provide a generalization in Section~\ref{SEC:Generalization}. 
Section~\ref{SEC:Relation to literature scc} illustrates connections to related literature. 
\hyperref[conclusion]{Conclusions} presents final remarks. 
Proofs and counterexamples are collected in the \hyperref[APP]{Appendix}.


\section{Stochastic choice with categorization} \label{SEC:Model}

We adopt the usual setup for stochastic choices from a finite set of options.
Let $X$ be a finite set of alternatives of cardinality $n \geqslant 3$, and denote $\X = 2^X \setminus \{\varnothing\}$ the collection of its nonempty subsets, referred to as `menus'.
The decision maker chooses one alternative from each menu, possibly randomizing, so that a probability of choice can be associated to each item, for each menu. 
This is represented by a stochastic choice.

\begin{definition} \label{DEF: stochastic choice function}
	A \textbf{stochastic choice} is a map $p \colon X \times \X \to [0,1]$ such that $p(a, A) = 0$ for all $a \notin A$, and $\displaystyle \sum_{a \in A} p(a, A) = 1$.
	If for all $A \in \X$, $a \in A$, $p(a, A) > 0$, then $p$ is \textbf{positive}.
\end{definition}

We interpret $p(a, A)$ as the probability item $a$ is chosen from menu $A$. 
For $C \subseteq A$, we define $p(C, A) \ = \ \displaystyle \sum_{a \in C} p(a, A)$ as the probability that an item in $C$ is chosen within $A$.
To simplify notation, we omit curly brackets in the argument of the choice function, e.g., we write $p(a, abc)$ in place of $p(a, \{a,b,c\})$.
Choice $p$ is the primitive in our model, and is assumed to be observable. 
Hereafter, we only focus on positive stochastic choice. \\

Now, partition $X$ in indexed subsets $\{X_i\}_{i \in I}$. 
The partition is induced by a surjective map $\pi \colon X \to I$, which assigns each item to a \textit{class} $i \in I$.
We say that $a$ belongs to class $i$ if $a \in X_i$ ($\pi(a) = i$). 
For each $i \in I$, let $\X_i = 2^{X_i} \setminus \{\varnothing\}$ be the collection of nonempty subsets of $X_i$.
Similarly, $\I = 2^I \setminus \{\varnothing\}$. 
We model a two-step choice process: 
the agent first picks a class, then an item within it.
In both choices, the agent randomizes among options at disposal.
This results in a \emph{stochastic choice with categorization} (\scc). 

\begin{definition} \label{DEF: resolvable functions}
	A positive stochastic choice $p \colon X \times \X \to [0, 1]$ is \textbf{with categorization} (\scc) if there exist 
	\begin{itemize}
		\item a partition of $X$ in classes $\{X_i\}_{i \in I}$, induced by a surjective map $\pi \colon X \to I$; 
		\item a positive stochastic choice on indices $\omega \colon I \times \I \to [0,1]$;
		\item a family $\{\sigma_i\}_{i \in I}$ of positive stochastic choices on classes, $\sigma_i \colon X_i \times \X_i \to [0,1]$;
	\end{itemize}
	such that, for all $A \in \X, \ i \in \pi(A), \ a \in A \cap X_i$,
	$$
	p(a, A) \ = \ \omega(i, \pi(A)) \ \sigma_i(a, A \cap X_i).
	$$
\end{definition}

We write $(p, \{X_i\}_{i \in I})$ to denote a stochastic choice $p$ with categorization with respect to partition $\{X_i\}_{i \in I}$.\\ 

In Definition~\ref{DEF: resolvable functions}, function $\pi$ represents the classification process, which is defined to be consistent across menus. 
Then, choice $\omega$ gives the probability of choosing a class in $\pi(A)$.
Note that, by definition, $\omega(i, \pi(A)) = \omega(i, \pi(B))$ for any $A, B$ such that $\pi(A) = \pi(B)$. 
In this sense, the probability of choosing any class is independent of its composition.
Finally, $\sigma_i(a, A \cap X_i)$ is the probability $a$ is chosen among items in its class.\\

Choices $\omega$ and $\{\sigma_i\}_{i \in I}$ have a natural relation to $(p, \{X_i\}_{i \in I})$:
\begin{itemize}
	\item $p(A \cap X_i, A) = \displaystyle \sum_{a \in A \cap X_i} p(a, A) \ = \ \displaystyle \sum_{a \in A \cap X_i} \sigma_i(a, A \cap X_i) \ \omega(i, \pi(A)) \ = \ \omega(i, \pi(A))$;
	\item for any $A \subseteq X_i, \quad p(a, A) \ = \ \sigma_i(a, A)$.
\end{itemize}

The decomposition $p(a, A) \ = \ \omega(i, \pi(A)) \ \sigma_i(a, A \cap X_i)$ allows to see that $\sigma_i(a, A \cap X_i)$ simply is the probability of choosing $a$ conditional on choosing in $X_i$.  
This conditional probability remains unaffected by any changes to the other classes in the menu.

All stochastic choices are trivially \scc{} with respect to the trivial partition $P = \{\{X\}\}$ and the $n$-partition of $X$, containing all singletons. 
Hereafter, we focus on cases in which the partition is significant. 

\begin{definition}
	A \scc{} $(p, \{X_i\}_{i \in I})$ is non-degenerate if $1 < |I| < n$.
\end{definition}

A stochastic choice can be non-degenerate \scc{} with respect to multiple partitions. 
Denote $\mathcal{P}$ the set of all such partitions.
For two distinct $P, P^\prime \in \mathcal{P}$, we say that $P$ is \textit{coarser} than $P^\prime$, denoted $P \succ P^\prime$, if every element of $P^\prime$ is a subset of some element in $P$. 
Formally, let $P = \{X_i\}_{i \in I}$ and $P^\prime = \{X_j\}_{j \in J}$. 
We have $P \succ P^\prime$ if for all  $j \in J$ there exists $i \in I$ such that $X_j \subseteq X_i$. 
Relation $\succ$ is a partial order on $\mathcal{P}$, and the set of the coarsest partitions is $\max(\mathcal{P}, \succ) = \{P \in \mathcal{P} \colon \nexists P^\prime \text{ such that } P^\prime \succ P\}$. 

\begin{restatable}[]{prop}{uniquenessofcoarsestpartition}
	There exists a unique coarsest partition $\{X_i\}_{i \in I} \in \mathcal{P}$. 
	Also, any other partition $\{X_j\}_{j \in J} \in \mathcal{P}$ is such that, for any $j \in J$, $X_j \subseteq X_i$ for some $i \in I$.
\end{restatable}

The proof (see Appendix~\ref{PROOF:uniquenessofcoarsestpartition}) establishes that 
finer partitions can only be obtained by breaking down a class into subclasses, which is only possible when some choice $\sigma_i$ on a class $X_i$ is itself non-degenerate \scc. 
If we maintain that such partitions mirror a specific form of similarity relation among items, the existence of a partition $P \in \mathcal{P}$ with $X_i \in P$ posits the existence of a \textquoteleft similarity threshold\textquoteright, for which all items in $X_i$ are sufficiently alike and all external items are not. 
A finer partition corresponds to a more stringent threshold applied within $X_i$. 
Our characterization outlines the behavior revealing that such structures can be recovered. 


The next examples show a natural instance of the model, and a dynamic interpretation, respectively.

\begin{example}[\it Two-step elimination by aspects]
	We formalize a two-step elimination by aspects \citep{Tversky1972}. 
	In the \textit{elimination by aspects} (\textsf{EBA}) model, options are perceived as bundles of attributes, and the agent selects the final choice by sequential elimination of items that lack attributes of interest. 
	Aspects are employed in a random order, with a probability of use at each elimination stage that depends on their value, so that the choice is stochastic.   
	For each item $x \in X$, denote $x^\prime$ the set of distinctive aspects present in $x$.
	Also, denote $A^\prime$ the set of distinctive aspects possessed by some alternative in menu $A$.
	Conversely, given a distinctive aspect $\epsilon$, $A_{\epsilon}$ is the set of items that present $\epsilon$, that is, $A_{\epsilon} = \{x \in A \colon \epsilon \in x^\prime\}$.
	A stochastic choice is \textsf{EBA} if there exists a positive function $u \colon X^\prime \to \mathbb{R}$, such that, for all $x \in A \subseteq X$
	$$
	\displaystyle p(x, A) \ = \ \frac{ \displaystyle \sum_{\epsilon \in x^\prime} u(\epsilon) \ p(x, A_\epsilon)}{\displaystyle \sum_{\delta \in A^\prime} u(\delta)}. 
	$$
	
	In a nutshell, the probability that item $x$ is chosen from $A$ depends on the value of its distinctive aspects, as captured by function $u$, and on the probability that $x$ stands out among alternatives sharing an aspect of interest, as given by $p$ itself. 
	
	Let us assume the agent's decision process for the choice of a car is tiered. 
	A set $M \subset X^\prime$ of macroscopic aspects, such as different price ranges and dimensions, is evaluated before set $N$ of remaining attributes, such as degrees of efficiency and colors, is considered.
	At each of the two steps, \textsf{EBA} is applied. 
	Let $x$ possess aspects in $M_x \subseteq M$ and $N_x \subseteq N$. 
	Similarly, for any menu $A$, let $M_A = M \cap A^\prime$ and $N_A = N \cap A^\prime$.
	The resulting stochastic choice is given by
	$$
	\displaystyle p(x, A) \ = \ \frac{\displaystyle \sum_{\epsilon \in M_x} u(\epsilon)}{\displaystyle \sum_{\delta \in M_A}u(\delta)} \ \cdot \ \frac{\displaystyle \sum_{\epsilon \in N_x} u(\epsilon) \ p(x, B_\epsilon)}{\displaystyle \sum_{\delta \in N_B} u(\delta)}, 
	$$
	where $B = \{y \in A: M_y = M_x\}$.
	The first ratio gives the probability that items presenting the combination of aspects $M_x$ survive the initial round of eliminations, and the second ratio is the probability that $x$ is selected by sequential elimination of surviving items.
	Combinations of macroscopic aspects specify classes, e.g., the class of compact cameras. 
	After the first stage,  only one class remains. 
	Further, the first ratio only depends on which other classes are available, not on their cardinality nor the identity of available options; the second ratio is solely determined by which items are available within the class of interest.  
	Clearly, the process is an instance of choice with categorization. 
\end{example}

The composition of \textsf{EBA} choices does not yield an \textsf{EBA} stochastic choice.\footnote{\textsf{EBA} choices belong to the class of Random Utility Models. We show that the choice resulting from this composition still belongs to the class (see Theorem \ref{THEO: RUM}).} 
Similarly, any multi-step choice may be composite and therefore lack a reliable representation by a single model. 
Our analysis investigates the conditions for decomposing an observed stochastic choice into independent subchoices.
This decomposition is instrumental for identifying applicable models for each subchoice.

\begin{example}[\it Dynamic choice]
	The structure of a stochastic choice with categorization can be represented by a decision tree, as defined by \cite{KrepsPorteus1979}. 
	For any given menu $A \in \X$, the first node corresponds to the problem of selecting a class in $\pi(A)$. 
	At this stage, the agent randomizes according to a distribution given by $\omega$. 
	The choice of a class leads to a new node, where an alternative within the class must be chosen. 
	Randomization at this stage is governed by $\sigma_i$, that applies to $A \cap X_i$, for the relevant class $i$. 
	The model imposes that each final node, each representing a distinct alternative in $A$, is only reachable by a single node (class), and especially that this structure is replicated across all menus. 
	This induces the required partition structure for $X$. 
	Also, the definition of choice $\omega$ entails a form of myopia, for which the specific composition of classes is irrelevant at the first stage. 
	This is akin to choosing to pursue a university degree based solely on the immediate satisfaction of studying certain subjects, neglecting to evaluate the ultimate alternatives, that is, the attainable job positions, that are contingent upon completing that course of study.
\end{example}

In this dynamic setting, we interpret the stochastic choice as an expression of deliberate randomization by a single individual. 
Given the limited observability for such instances, choice probabilities for the full domain could only be defined through thought experiments with counterfactual menus.
We still consider $p$ the primitive, and its decomposition is revealing of the emergent tree structure and of this form of myopia.\\

Counterexample~\ref{EX:nonresolvable} in Appendix~\ref{SEC:a nonSSC} exhibits a choice which is not \scc, showing the model is falsifiable.


\section{A characterization} \label{SEC:Characterization}

Let $p \colon X \times \X \to [0,1]$ be a positive stochastic choice. 
To deduce whether $p$ is \scc, we define the choice behavior revealing the existence of a category in $X$. 

\bigskip

\begin{definition} \label{DEF:fun-proper category}
	A \textbf{category} is a set $C \in \X$ such that the following properties hold:\\
		
		\cind \label{cind}: for all $S \subseteq C$, $a, b \in S$, $E \subseteq X \setminus C,$
		$$  \displaystyle \frac{p(a, S)}{p(b, S)} \ = \ \frac{p(a, S \cup E)}{p(b, S \cup E)};$$
		
		\cneu \label{cneu}: for all $E \subseteq X \setminus C, \ x \in E, \ S \subseteq C, \ S \neq \varnothing$,
		$$ p(x, S \cup E) \ = \ p(x, C \cup E).$$
\end{definition}

Set $X$ and singletons are trivially categories in $X$. 
We say that $C$ is a \textbf{non-trivial} category if $1 < |C| < |X|$.\\


The interpretation of a \emph{category} relies on two assumptions.
First, external alternatives are sufficiently distinct to not convey any relevant information that would alter the assessment of items within the category.
Thus, \cind{} posits that the relative choice probabilities for items within a given category remain invariant to the presence of external items. 
Second, the competition for any given alternative is confined to other items within its category, that are therefore perceived as substitutes. 
This form of substitutability is captured by \cneu, requiring that the addition of items of a represented category does not affect the probability any external item is chosen. 
In Appendix~\ref{SEC:independence of the axioms}, we show that \cind{} and \cneu{} are independent.

It is easy to see that the union of two categories is \textit{not} a category, in general. 

\begin{restatable}[]{prop}{categoriesintersection}
	Any two distinct non-trivial categories $C, D$ must be either disjoint or nested. Namely, $C \cap D \neq \varnothing$ implies that $C \subseteq D$ or $D \subseteq C$.
\end{restatable}

See Appendix~\ref{SEC: categorical structures} for a proof. 
These results delineate the admissible categorical structures. 
If we interpret categorization as related to similarity relations among alternatives, we learn that items within a category are sufficiently alike, and the degree of similarity can be further specialized to isolate a subset that is itself a category. 
The main theorem can now be stated.

\bigskip

\begin{restatable}[]{theorem}{resolvability}	\label{THEO: characterization resolvability}
	A positive stochastic choice $p \colon X \times \X \to [0, 1]$ is a non-degenerate \scc{} if and only if there exists a non-trivial category in $X$.
\end{restatable}

The proof is constructive. \cind{} fixes a stochastic choice on each class ($\{\sigma_i\}_{i \in I}$), while \cneu{} allows to specify a well-defined choice among classes ($\omega$) (see Appendix~\ref{PROOF:categorization}).
Disjoint non-trivial categories are classes in the partitions that define non-degenerate \scc.
In particular, the coarsest partition contains all maximal non-trivial categories in $X$. 

\begin{corollary}
	Let $\mathcal{P}$ be the set of partitions for which $p$ is non-degenerate \scc, and let $\mathcal{C}^M$ be the collection of non-trivial maximal categories in $X$. 
	The unique coarsest partition $\{X_i\}_{i \in I} \in \mathcal{P}$ is such that for all $C \in \mathcal{C}^M$, $C = X_i$, for some $i \in I$.
\end{corollary}

\section{A representation}	\label{SEC:deterministic resolvability}

Stochastic choices can represent aggregate choice frequencies in a population, when a single observation per menu is available for each individual. 
In this section, we single out the features of a population that generates a \scc. 
First, we introduce some useful definitions.\\

A (deterministic) choice function on a set $X$ is a map $c \colon \X \to X$ such that $\varnothing \neq c(A) \in A$, for all $A \in \X$.
We find that all \scc{} are generated by a probability distribution on a set of \textit{resolvable} choice functions, as modeled in \cite{Cantoneetal.2020}.\footnote{The original definition refers to choice correspondences. Here we report a formulation for choice functions.}

\begin{definition} \label{DEF: deterministic resolvable choice}
	A choice $c_X \colon \X \to X$ is \textbf{resolvable} with respect to a partition $\{X_i\}_{i \in I}$, induced by a surjective map $\pi \colon X \to I$, if there exist a choice on indices $c_I \colon \mathcal{I} \to I$, and a family of choices on classes $\{c_i\}_{i \in I}$, such that $\displaystyle c_X(A) = c_i(A \cap X_i)$, with $i = c_I(\pi(A))$.
\end{definition}	

Thus, given a partition indexed by $I$, for any menu $A$, first an index $i$ is picked according to $c_I \colon \I \to I$, such that 
$
c_I(\pi(A)) \ = \ \pi(c_X(A)); 
$
then an alternative in $A \cap X_i$ is chosen by $c_i \colon \X_i \to X_i$.
No further form of consistency is required. 

For a partition $\{X_i\}_{i \in I}$, denote $\mathcal{R}$ the set of choices resolvable with respect to $\{X_i\}_{i \in I}$.
Each function in $\mathcal{R}$ is an instance of individual choice behavior. 
By defining a probability distribution $Q$ on $\mathcal{R}$, we can generate a stochastic choice representing the occurrence of each choice outcome. 
In particular, the following result holds.

\begin{restatable}[]{theorem}{resolvabilitypopulation} \label{THEO:resolvabilitypopulation}
	Let $(p, \{X_i\}_{i \in I})$ be non-degenerate \scc{} and $\mathcal{R}$ the set of choices resolvable with respect to $\{X_i\}_{i \in I}$. 
	There exists a probability distribution $Q$ on $\mathcal{R}$ such that:
	$$
	p(a, A) \ = \ \sum_{c_X : c_X(A) = a} Q(c_X).
	$$
\end{restatable}

Thus, any \scc{} is generated by a population of individuals who share the same classification structure. 
However, the converse does not hold, as illustrated by the following counterexample.

\begin{counterexample} \label{EX:nonresolvable population}
	Let $X = \{a, b, c\}$ be the set of available items, and consider the set $\mathcal{R}$ of deterministic choices on $X$ resolvable with respect to partition $P = \{X_1, X_2\}$ with $X_1 = \{a, b\}$ and $X_2 = \{c\}$. 
	Each $c_X \in \mathcal{R}$ has components $c_I$, with $c_I(\{1, 2\}) \in \{1, 2\}$, and $c_1, c_2$, with $c_1(\{a, b\}) \in \{a, b\}$. 
	It follows that $\mathcal{R}$ counts four distinct choices. 
	The table below reports each choice (for each menu, the chosen item is underlined), and the probability assigned by a distribution $Q$ on $\mathcal{R}$.
	At the bottom of the table is reported the probability each item is chosen in each menu, that is, the stochastic choice on $X$ resulting from $Q$.
	
	\begin{center}
		\begin{tabular}{ c c c c c | c }
			&	&	$c_X$ &	 	 &	 	& $Q(c_X)$ \\
			& $a$ $b$ $\underline{c}$ & $a$ $\underline{c}$ & $\underline{a}$ $b$ & $b$ $\underline{c}$ & $\frac{1}{5}$ \\
			& $a$ $b$ $\underline{c}$ & $a$ $\underline{c}$ & $a$ $\underline{b}$ & $b$ $\underline{c}$ & $\frac{2}{5}$ \\
			& $\underline{a}$ $b$ $c$ & $\underline{a}$ $c$ & $\underline{a}$ $b$ & $\underline{b}$ $c$ & $\frac{1}{5}$ \\
			& $a$ $\underline{b}$ $c$ & $\underline{a}$ $c$ & $a$ $\underline{b}$ & $\underline{b}$ $c$ & $\frac{1}{5}$ \\
			\hline \hline
			$p(a, \cdot)$ & $\frac{1}{5}$ & $\frac{2}{5}$ & $\frac{2}{5}$ & $\diagup$ & \\ 
			$p(b, \cdot)$ & $\frac{1}{5}$ & $\diagup$ & $\frac{3}{5}$ & $\frac{2}{5}$ & \\ 
			$p(c, \cdot)$ & $\frac{3}{5}$ & $\frac{3}{5}$ & $\diagup$ & $\frac{3}{5}$ &   
		\end{tabular}
		\captionof{table}{Stochastic choice $p$ generated by $Q$ on $\mathcal{R}$.}
	\end{center}
	
	Notice that $\displaystyle \frac{p(a, ab)}{p(b, ab)}  =  \frac{2}{3} \ \neq \ \frac{p(a, abc)}{p(b, abc)}  =  1$.
	\bigskip
	This is a violation of \cind, entailing that $X_1 = \{a, b\}$ is not a category, and $p$ is not \scc{} with respect to partition $P$.
	Indeed, there exists no partition for which $p$ is non-degenerate \scc.
\end{counterexample}
To see why this violation occurs, consider what follows.
If $a, b$ and $c$ are available, only individuals with $c_I(12) = 1$ choose within $X_1 = \{a, b\}$, so that $\frac{p(a, abc)}{p(b, abc)} \ = \ 1$. 
However, when $c$ is excluded, individuals with $c_I(12) = 2$ necessarily deviate to $a$ or $b$, in a proportion that could modify the probability ratio \big($\frac{p(a, ab)}{p(b, ab)}  =  \frac{2}{3}$\big).
Hence, a probability distribution on the set of deterministic resolvable choices need not generate a stochastic choice with categorization.

The additional property required for $Q$ to generate non-degenerate \scc{} ensures that violations such as those in Counterexample~\ref{EX:nonresolvable population} are forbidden. 
In particular, for any $c_X$ it must be
$$
\displaystyle Q(c_X) \ = \ Q(\{c^\prime_X \colon c_I^\prime = c_I\}) \prod_{i \in I} Q(\{c^\prime_X \colon c^\prime_i = c_i\}),
$$
which again requires a form of independence for the induced distribution on choices within classes. 
This entails that, in menus where additional classes are made available, agents change their choice while preserving probability ratios for the original items.

Finally, the following theorem characterizes stochastic choices generated by a population of deterministic resolvable choices.

\begin{restatable}[]{theorem}{population} \label{THEO:population}
	Let $p \colon X \times \X \to [0, 1]$ be a positive stochastic choice, let $\pi \colon X \to I$ induce a partition $\{X_i\}_{i \in I}$ of $X$, and let $\mathcal{R}$ collect deterministic choices resolvable with respect to $\{X_i\}_{i \in I}$.
	The following are equivalent:
	\begin{itemize}
		\item[(1)] $\text{for all } A, B \in \X \text{ s.t.\ } \pi(A)  =  \pi(B), \quad p(A \cap X_i, A) \ = \ p(B \cap X_i, B)$,
		\item[(2)] $p$ is represented by a probability distribution $Q$ on $\mathcal{R}$, 
		$$
		p(a, A) \ = \ \sum_{c_X : c_X(A) = a} Q(c_X).
		$$
	\end{itemize}
\end{restatable}

A proof is provided in Appendix~\ref{SEC:a population}. 
Here, we point out that condition \textit{(1)} is implied by \cneu{} applied to each set in $\{X_i\}_{i \in I}$. 
Expectedly, \scc{} is a subset of the set of stochastic choices generated by a population of deterministic resolvable choices. 
More interestingly, the two-step procedure in deterministic choice emerges in \cneu{} in the induced stochastic choice.


\section{A generalization} \label{SEC:Generalization}

\subsection{Stochastic choice with weak categorization}

For a \scc{}, the choice among classes is represented by a menu-independent stochastic choice, $\omega$ (Definition \ref{DEF: resolvable functions}). 
In this section, we lift this restriction and 
introduce choices \emph{with weak categorization} (\scwc).

\begin{definition} \label{DEF:weakcategorization}
	A positive stochastic choice $p \colon X \times \X \to [0,1]$ is \textbf{with weak categorization} (\scwc) if there exist
	\begin{itemize}
		\item a partition of $X$ in subsets $\{X_i\}_{i \in I}$, induced by a surjective map $\pi \colon X \to I$;
		\item a family $\{\omega_A\}_{A \in \X}$ of functions $\omega_A(\cdot, \pi(A)) \colon I \to [0,1]$ such that
		\begin{itemize}
			\item $\omega_A(i, \pi(A)) = 0 \ \Leftrightarrow i \notin \pi(A)$, 
			\item $\displaystyle \ \sum_{i \in \pi(A)} \omega_A(i, \pi(A)) \ = \ 1$;
		\end{itemize}
		\item a family $\{\sigma_i\}_{i \in I}$ of positive choices $\sigma_i \colon X_i \times \X_i \to [0,1]$; 
	\end{itemize}
	such that for all $A \in \X$, $i \in \pi(A)$, $a \in A \cap X_i$,
	$$ 
	p(a, A) = \ \omega_A(i, \pi(A)) \ \sigma_i(a, A \cap X_i). 
	$$ 
	A \scwc{} $(p, \{X_i\}_{i \in I})$ is non-degenerate if $1 < |I| < n$.
\end{definition}

Clearly, \scc{} is obtained as a special case under the restriction that $\omega_A(i, \pi(A)) \ = \ \omega_B(i, \pi(B))$ for all $A, B \in \X$ such that $\pi(A) = \pi(B)$.  
As for \scc, a stochastic choice can be with weak categorization with respect to multiple partitions. 
Nonetheless, in virtue of the greater flexibility of the model, multiple coarsest partitions can be outlined. 

Further, by employing a family of choices $\{\omega_A\}_{A \in \X}$, each indexed by a menu, Definition~\ref{DEF:weakcategorization} allows displays of erratic behavior in choices \textit{among} and \textit{within} classes, as long as items are consistently classified and conditional choices $\{\sigma_i\}_{i \in I}$ are fixed. 
This feature informs the structure of the model. 
Instances of menu-dependence in the choice of a class include several well-documented phenomena, such as choice overload \citep{Chernevetal2015, IyengarLepper2000}, associationistic salience \citep{LiTang2016} and reference dependence \citep{Barbos2010}.

Similar to \scc, stochastic choices with weak categorization can be modeled using a decision tree. 
Suitable specifications of choices $\omega_A$ can accommodate for sophistication in the choice of a class. 
For instance, in Section \ref{SEC:Relation to literature scc}, we show that Nested Logit is a special case of \scwc.

\subsection{A characterization of \scwc}

To assess whether $p$ is \scwc{}, we investigate the existence of a \textit{weak category}, a generalization of Definition~\ref{DEF:fun-proper category}.

\begin{definition} \label{DEF:general category}
	A \textbf{weak category} is a set $W \in \X$ such that the following holds:\\ 
	
	\cind: for all $S \subseteq W$, $a, b \in S$, $E \subseteq X \setminus W$,
	\begin{equation*} 
		\displaystyle \frac{p(a, S)}{p(b, S)} \ = \ \frac{p(a, S \cup E)}{p(b, S \cup E)}.
	\end{equation*}
	A weak category $W$ is \emph{non-trivial} if $1 < |W| < |X|$.
\end{definition}

A set of items constitutes a \emph{weak category} if the ratios of choice probabilities are not altered by the introduction of external items. 
This framework acknowledges that choices may violate \cneu{} in response to contextual changes, possibly conveyed by a change in the number or identity of available items. 
For instance, a decision-maker observing a wide selection of mirrorless cameras might infer that the retailer specializes in this category, and conclude that high-quality options are available. 
This, in turn, increases the likelihood of choosing within this category. 

This approach allows more flexible categorical structures: weak categories are not required to be disjoint or nested, and the intersection of any two weak categories is also a weak category (see Appendix~\ref{SEC: categorical structures} for a proof). 
The union of any two weak categories is not a category, in general.

\begin{restatable}[]{prop}{weakcategoriesintersection}
	If $W$ and $T$ are weak categories and $W \cap T \neq \varnothing$, then $W \cap T$ is a weak category.
\end{restatable}

Classification structures in this case admit a specific form of non-transitive similarity between items: it may be the case that there is some item $a \in W \cap T$, so that $a$ is sufficiently similar to item $b \in W$ and $c \in T$, however $b$ and $c$ are not similar to the same degree.
This is forbidden for categories as in Definition~\ref{DEF:fun-proper category}.

The existence of a weak category characterizes \scwc{}.

\begin{restatable}[]{theorem}{generalresolvability} \label{THEO: characterization general resolvability}
	A positive stochastic choice $p \colon X \times \X \to [0, 1]$ is non-degenerate \scwc{} if and only if there exists a non-trivial weak category in $X$.
\end{restatable}

Disjoint weak categories in $X$ are classes in a partition for which $p$ is \scwc, \cind{} fixes choices $\{\sigma_i\}_{i \in I}$, choices $\{\omega_{A}\}_{A \in \X}$ are then easily retrieved (see proof in Appendix~\ref{PROOF: weak categorization}). 


\begin{example}
	Suppose an agent is only able to compare and rank similar items. 
	Let $G$ be such a subset of comparable items in $X$, which can be ranked in a linear order.
	In particular, when choosing within $G$, one linear order in $\mathcal{L}_G$ realizes with a certain probability, and is applied by the agent. 
	In virtue of this procedure, choice restricted to $G$ has a random utility representation. 
	We say that a stochastic choice $p \colon X \times \X \to [0,1]$ is \emph{locally rationalizable} at $G \subseteq X$ if there exists a probability distribution $q$ on the set of linear orders $\mathcal{L}_G$ on $G$ such that, for all $x \in G$, $A \subseteq X$, 
	$$
	p(x, A) \ = \ p(G \cap A, A) \ \ q(\{\ell_G \in \mathcal{L}_G \colon \max(A \cap G, \ell_G) = x\}).
	$$  
	It is easy to prove the following result (see Appendix \ref{local rationalizability}). 
	
	\begin{restatable}[]{prop}{localrationality}
		If $p$ is locally rationalizable at $G$ such that $1 < |G| < |X|$, then $G$ is a non-trivial weak category.
	\end{restatable}
	\noindent It follows from Theorem~\ref{THEO: characterization general resolvability} that a locally rationalizable $p$ is a non-degenerate \scwc. 
\end{example}


\section{Related literature} \label{SEC:Relation to literature scc}

In a deterministic setting, sequential procedures have been modeled as a result of choice over menus in anticipation of future preferences \citep{Kreps1979, GulPesendorfer2001, Heydari2020}, bounded rationality \citep{Masatlioglu2016, LlerasMasatliogluNakajima2017, Fricketal2019}, preliminary screening steps \citep{ManziniMariotti2007, ApesteguiaBallester2013, Tyson2013}), or categorization \citep{Barbos2010, ManziniMariotti2012, Aguiar2017}. 
In all these instances, the agent focuses on a subset of available items, and picks a selection within it.
This work is especially close in spirit to \cite{Cantoneetal.2020}, recently extended to a probabilistic setting in \cite{CarpentiereDoignon2025}.  
The latter, developed independently, generalizes our definition and characterization of \scc{} to the broader framework of stochastic choice correspondences. 
To the best of our knowledge, our considerations on partition structures and on the generating population of a \scc{}, as well as the definition and characterization of \scwc{}, have not been previously addressed in the literature. 
Hereafter, we contrast our models to Luce's model and Random Utility models to clarify the implications of our axioms, before briefly delving into other related models. 

\subsection{On Luce's IIA}

We start by noting that non-degenerate \scc are \textit{not} Luce model choices, defined as follows \citep{Luce1959}.

\begin{definition}
	A positive stochastic choice $p \colon X \times \X \to [0,1]$ is Luce model if there exists a utility function $u \colon X \to \mathbb{R}_{++}$ such that $p(a, A) = \frac{u(a)}{\sum_{x \in B} u(x)}$, for all $A \in \X$.
\end{definition}

Now let $(p, \{X_i\}_{i \in I})$ be a non-degenerate \scc{}. 
Without loss of generality, consider a set $A = \{a, b, c\} \subseteq X$ such that $a, b \in X_i$ and $c \in X_j$ for some $i, j \in I$. 
We check the existence of a utility function $u \colon X \to \mathbb{R}_{++}$ such that $p(a, B) = \frac{u(a)}{\sum_{x \in B} u(x)}$, $B \subseteq A$.
Since $(p, \{X_i\}_{i \in I})$ is \scc, $p(a, ac) \ = \ p(b, bc) \ = \ p(ab, abc) \ = \ \omega(i, ij)$. 
Thus, it must be 
$$
p(a, ac) = p(ab, abc) \ \Leftrightarrow \ \frac{u(a)}{u(a) + u(c)} =  \frac{u(a) + u(b)}{u(a) + u(b) + u(c)} \ \Leftrightarrow \ u(b) u(c)= 0,
$$
which contradicts the assumption of strict positivity for $u$. 
Since $X_i$ is a category, we can turn to Definition \ref{DEF:fun-proper category} and see that this contradiction is determined by \cneu. 
Luce model is famously characterized by \IIA. 

\begin{definition}
	A positive stochastic choice $p$ satisfies \IIA{} if $\dfrac{p(a, ab)}{p(b, ab)} = \dfrac{p(a, A)}{p(b, A)}$ for all $a, b \in X$ and $A \in \X$. 
\end{definition}

Clearly, \IIA{} is equivalent to imposing that any pair $a,b$ satisfies \cind, and is therefore a weak category.

Stochastic choice with categorization has an intersection with Nested Stochastic Choice.
We report the definition as introduced in \cite{Kovach2022}.


\begin{definition}[\textsf{NSC}] \label{def-NSC}
	A stochastic choice $p \colon X \times \X \to [0,1]$ is a \textit{nested stochastic choice} (\nsc) if there are a partition $\{X_i\}_{i \in I}$, a utility function $u \colon X \to \mathbb{R}_{++}$, and a nest utility function $v \colon \bigcup_{i \in I} 2^{X_i} \to \mathbb{R}_+$ with $v(\varnothing) = 0$ such that, for any $A \in \X$ and $a \in A \cap X_i$, 
	$$
	\displaystyle p(a, A) \ = \ \frac{v(A \cap X_i)}{\displaystyle \sum_{j \in I} v(A \cap X_j)} \ \frac{u(a)}{\displaystyle \sum_{b \in A \cap X_i} u(b)}.
	$$
\end{definition}

Thus, item $a \in A \cap X_i$ is chosen within $A$ if \textit{(i).}\ nest $A \cap X_i$ is chosen in $A$, with probability $\frac{v(A \cap X_i)}{\sum_{j \in I} v(A \cap X_j)}$; \textit{(ii).}\ $a$ is chosen in $A \cap X_i$, with probability $\frac{u(a)}{\sum_{b \in A \cap X_i} u(b)}$.
%
A \nsc{} $(p, \{X_i\}_{i \in I})$ is a \scc{} only if, for all $ A, B \in \X \text{ s.t.\ } \pi(A) = \pi(B), \text{ for all } i \in I, \ \ v(A \cap X_i) \ = \ v(B \cap X_i)$. 
This condition allows $\omega(i, \pi(A)) = \frac{v(A \cap X_i)}{\sum_{j \in I} v(A \cap X_j)}$ to be well-defined.
Nested Logit is a special case of \nsc{} (see \cite{Kovach2022}), such that, for all $A \in \X$, $v(A \cap X_i) \ = \ \big( \sum_{a \in A \cap X_i} u(a)\big)^{\lambda_i}$, for a set of strictly positive real numbers $\{\lambda_i\}_{i \in I}$. 
It is easily shown that there must be some $A, B \in \X \text{ s.t. } \pi(A) = \pi(B)$ such that $\big( \sum_{a \in A \cap X_i} u(a)\big)^{\lambda_i} \ \neq \ \big( \sum_{b \in B \cap X_i} u(b)\big)^{\lambda_i}$, so that Nested Logit is an instance of \nsc{} which is not non-degenerate \scc.  
Conversely, any nest is a \textit{weak category}, and the link to \scwc{} can be made apparent by 
letting, for all $A \in \X, \ a \in A \cap X_i$,
$$
\omega_A(i, \pi(A)) = \frac{v(A \cap X_i)}{\sum_{i = 1}^K v(A \cap X_j)} \quad \text{ and } \quad \sigma_i(a, A \cap X_i) = \frac{u(a)}{\sum_{b \in A \cap X_i} u(b)}. 
$$ 
It follows naturally that Nested Logit is a \scwc. 

\begin{remark}[On \IIA{} and its weaker forms] 
	\cite{Luce1959} introduces \IIA{} as a conventional form of rationality: additional items are irrelevant for assessing choice probabilities, and therefore relative preference, for any pair $a$ and $b$. 
	\cite{Debreu1960}'s example shows how the axiom is too stringent, as it prevents common forms of substitutability.
	To accommodate for that, the characterization of \nsc{} relies on a notion of revealed similarity, defined as follows: $a \sim_p b$ if and only if $\dfrac{p(a,ab)}{p(b, ab)} = \dfrac{p(a,A)}{p(b, A)}$ for every $A$. 
	If this condition is not satisfied, $a$ and $b$ are revealed dissimilar, denoted $a \nsim_p b$. 
	The characterizing axiom, \textsc{ISA}, guarantees that $\sim_p$ is an equivalence relation, which therefore allows the partition of $X$ in nests: 
	$$
	a \sim_p x \text{ and } b \sim_p x \quad \text{ or } \quad a \nsim_p x \text{ and } b \nsim_p x \qquad \Rightarrow \qquad \dfrac{p(a,A)}{p(b, A)} = \dfrac{p(a,A \cup x)}{p(b, A \cup x)}. 
	$$
	Violations of \IIA{} are hence allowed, only to reveal substitutability patterns. 
	In particular, if \IIA{} holds for a pair of options, it may be an expression of rationality for items in the same nest, or an exhibition of incomparability for items in different nests. 
	Our notion of \cind{} further weakens \IIA, and is based solely on incomparability. 
	The extreme case in which all items are incomparable, that is, any pair is a weak category, artificially yields rationality (\IIA) as a by-product.  
	On the other hand, \cneu{} unequivocally pins down substitutability structures, and has no further ties with rational behavior. 
	This is indeed further acknowledged by our representation result (Theorem \ref{THEO:resolvabilitypopulation}). 
	This observation further suggests that to effectively disentangle substitutability in choice objects from rationality in decision-making, it is useful to explicitly formalize the initial stage of the two-step procedure.
\end{remark}

\subsection{On Random Utility models}

Next, we focus on the relation of \scc{} to the family of random utility models (\rum). 
A stochastic choice $p$ on $X$ admits a random utility representation \citep{BlockMarshack1960, Falmagne1978} if there exists a probability distribution $\mathbb{P}$ on the set $\mathcal{L}$ of linear orders\footnote{A linear order is a transitive, antisymmetric and complete (hence reflexive) binary relation.} on $X$, such that
\begin{itemize}
	\item $\mathbb{P}(L) \geqslant 0$ for all $L \in \mathcal{L}$;
	\item $\displaystyle \sum_{\mathcal{L}_{(a, A)}} \ \mathbb{P}(L) = p(a, A), \ \text{ with } \mathcal{L}_{(a, A)} \ = \ \{L \in \mathcal{L} \colon  a  L  b, \ \forall b \in A \setminus \{a\}\} $.
\end{itemize} 

If this is the case, we say that $p$ is rationalizable by \rum, or simply \rum.

Stochastic choices with categorization need not admit a random utility representation. 
Yet, rationalizability of the observable \scc{} can be linked to rationalizability of its components. 
The following result can be derived as a special case of Theorem 7.4 in \cite{CarpentiereDoignon2025}. 
Our proof is original.  

\begin{restatable}[]{theorem}{RUMstability} \label{THEO: RUM}
	A stochastic choice with categorization $(p, \{X_i\}_{i \in I})$ is \rum{} if and only if choices $\{\sigma_i \colon X_i \times \X_i \to [0, 1]\}_{i \in I}$ and $\omega \colon I \times \I \to [0, 1]$ are \rum{}.
\end{restatable}

In other words, rationalizability by \rum{} is inherited from $p$ by its components $\omega, \{\sigma_i\}_{i \in I}$; vice versa, it is passed on from $\omega$ and $\{\sigma_i\}_{i \in I}$ to their composition, $p$. 
In general, $\omega$ and $\{\sigma_i\}_{i \in I}$ may admit multiple \rum{} representations, hence the same holds for $p$. 
The proof (Appendix~\ref{PROOF:categorizationRUM}) shows that, for such choices, the support of any $\mathbb{P}$ rationalizing $p$ only contains linear orders which unequivocally rank classes, that is, which sequentially arrange items in the same class.

\subsection{On related stochastic choice models}

Stochastic choices with categorization also relate to two prominent random utility models explicitly accounting for similarity between items.
In \textit{elimination by aspects} (\eba) by \cite{Tversky1972}, options are perceived as bundles of aspects. 
The probability item $x$ is chosen from $A$ depends on how valuable its distinctive aspects are, and how likely $x$ stands out among alternatives sharing aspects of interest. 
According to the \emph{attribute rule} in \cite{GulNatenzonPesendorfer2014}, the probability an item $x$ is chosen depends on which attributes are exhibited in $x$, and, specifically, the intensity of their presence.
Attribute rules and \eba{} are not \scc{}, as they are in principle procedures for one-step choices.  
The \scc{} model shares, however, the interest in the implications of similarities among items for the choice process. 
We argue that defining similarity rigorously ex ante is difficult. 
Its perception may depend on shared properties and on the intensity of these features, as it is the case for attribute rules and \eba; alternatively, it may be goal-derived, or based on a pre-established classification structure. 
As the \scc{} model is agnostic on the motives for item classification, it could be applied to empirically uncover similarity criteria.
\scc{} also differs from \textit{sequential Luce model} \citep{Zhang2016}, where the agent is allowed a default option and \IIA{} must hold within a class.
Under the additional assumption that items within a class have equal probability of choice, stochastic choices by categorization reduce to \textit{Luce with replicas}, introduced by \cite{Faro2023}.
\textit{Luce with replicas} is also a particular case of \nsc, so that it lies in the intersection of \nsc{} and non-degenerate \scc.\\

The Associationistic Luce Model by \cite{LiTang2016} accounts for the phenomenon where an item's salience extends to its containing class.
First a class is chosen according to a Luce-type formula with associationistic salience as weights, then an item within it is selected according to Luce's rule.
The model is a special case of \scwc{}, with the restriction that choices $\{\sigma_i\}_{i \in I}$ are Luce model. 
Stochastic choices with weak categorization also have affinity to Perception adjusted Luce model (\textsf{PALM}) by \cite{EcheniqueSaitoTserenjigmid2018}.
Their model describes choices of an agent that considers items in a priority order. 
As a result, items are chosen according to \IIA{} only if evaluated simultaneously. 
Under the assumption that \IIA{} holds within each class, and that alternatives in the same class are noticed at the same time, by association, \scwc{} come conceptually close to \textsf{PALM}. 
However, technical differences (e.g., \textsf{PALM} contemplate a default option) prevent from drawing a formal comparison of the models.
Finally, we recognise similarities to the approach in General Luce model, introduced in \cite{Echeniqueetal2019}. 
The model entails a two-step choice. 
First, a total order is applied, to isolate maximal elements in the menu. 
Then a Luce's type formula is applied to break ties. 
In fact, the total order naturally induces a partition of items in equivalence classes. 
While these could be interpreted as categories, \scwc{} only refers to positive stochastic choices, so that the two models of choice do not intersect.


\begin{figure}[htbp]
	\centering
	\begin{tikzpicture}[
		outer/.style={draw=blue!20!black , thick, fill=blue!3},
		middle/.style={draw=green!60, thick, fill=green!3},
		inner/.style={draw=orange!70, thick, fill=orange!8},
		nested/.style={draw=red!60!black, thick, fill=red!2},
		core/.style={draw=purple!70!black, thick, fill=purple!5, minimum size=1.4cm},
		label/.style={font=\small\bfseries, text=black!80}
		]
		
		\fill[blue!1] (0,0) ellipse (5.8cm and 3.5cm);
		\fill[green!2] (-2.2,0) ellipse (2.1cm and 2.4cm);
		\fill[green!2] (2,0) ellipse (3.2cm and 2.4cm);
		
		\draw[blue!50!black, thick] (0,0) ellipse (5.8cm and 3.5cm);
		\draw[gray!60!black, thick] (-2.2,0) ellipse (2.1cm and 2.4cm);
		\draw[green!50!black!, thick] (2,0) ellipse (3.2cm and 2.4cm);
		
		\node[label, rounded corners=2pt, inner sep=3pt] 
		at (0,3.8) {n-deg \scwc};
		\node[label, rounded corners=2pt, inner sep=3pt] 
		at (-2.2,2.7) {n-deg \scc};
		\node[label, rounded corners=2pt, inner sep=3pt] 
		at (2,2.7) {\nsc};
		
		\draw[nested] (2.2,0) ellipse (2cm and 1.5cm);
		\node[label, rounded corners=2pt, inner sep=3pt] 
		at (2.2,1.8) {Nested Logit};
		
		\node[draw, circle, core] at (2.2,0) {};
		\node[label, rounded corners=2pt, inner sep=3pt] 
		at (2.2,-0.9) {Luce Model};
		
	\end{tikzpicture}
	\caption{Stochastic choices by (weak) categorization and related choice models. 
		Since degenerate cases of \scc{} and \scwc{} do not restrict behavior, we neglect them in this diagram.}
	\label{FIG1:insiemini}
\end{figure}
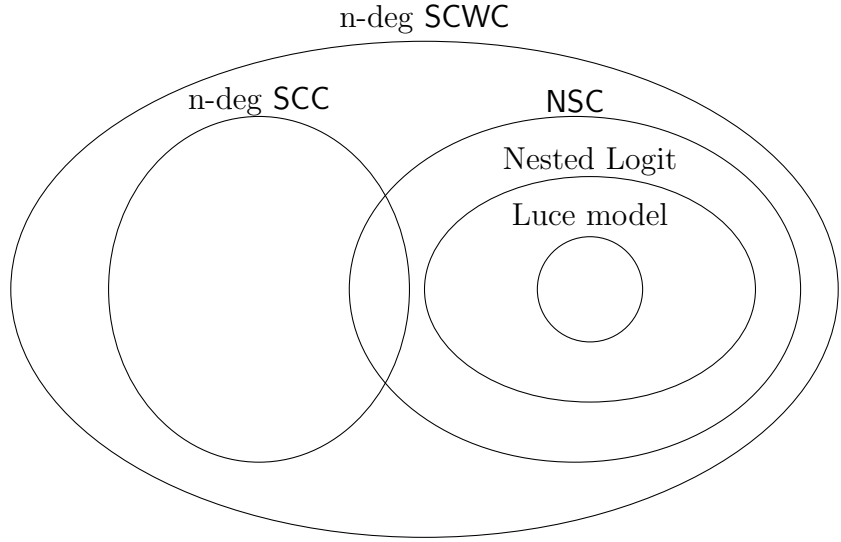


\section*{Conclusions} \label{conclusion}

This work defines and characterizes stochastic choices with categorization. 
The model is permissive, as no particular form of rational nor behavioral decision-making is assumed.
Yet, the definition is tight enough to discriminate choices that can be decomposed into independent sub-choices. 
We show that any \scc{} can be represented by a probability distribution on a set of deterministic choices which share a two-step structure with respect to the same partition. 
More precisely, the two-step procedure in deterministic choice behavior emerges in \cneu{} for the induced stochastic choice. 
We also introduce and characterize a generalization of the model. 
Our definition of a (weak) category sheds light on the implications of Luce's \IIA{} and its variations for choice behavior, and further suggests how modelling the first step of choice can prove useful to define relations of substitutability among items. 
Designing a statistical test to uncover categories from choice data falls outside the scope of this work, and remains a natural direction for future research.



\section*{APPENDIX} \label{APP}

\begin{appendices}
	
	\section{Counterexamples} 
	
	\subsection{Falsifiability}\label{SEC:a nonSSC}
	
	The next counterexample shows that even for a small set $X$, stochastic choices by categorization constitute a proper subset of all possible stochastic choices. 
	Hence, Definitions~\ref{DEF: resolvable functions} and \ref{DEF:weakcategorization} actually provide a discrimination criterion for choices.
	
	\begin{counterexample} \label{EX:nonresolvable}
		Let $X = \{a, b, c\}$.
		\begin{itemize}
			\item $p(a, ab) = \frac{1}{2}$; $\quad p(a, abc) = \frac{1}{8};$
			\item $p(b, bc) = \frac{1}{3}$; $\quad p(b, abc) = \frac{1}{6};$
			\item $p(c, ac) = \frac{1}{4}$; $\quad p(c, abc) = \frac{17}{24}.$
		\end{itemize}
		For $p$ to be \scc, at least one binary set has to be a \textit{category}. 
		Below we check whether $p$ satisfies the definition.
		\begin{itemize}
			\item $\{a, b\}$ : $p(ab, abc) = \frac{7}{24}; \ \ p(a, ab) \ p(ab, abc) = \frac{1}{2} \frac{7}{24} = \frac{7}{48} \neq p(a, abc)$
			\item $\{b, c\}$ : $p(bc, abc) = \frac{7}{8}; \ \ p(b, bc) \ p(bc, abc) = \frac{1}{3} \frac{7}{8} = \frac{7}{24} \neq p(b, abc)$
			\item $\{a, c\}$ : $p(ac, abc) = \frac{5}{6}; \ \ p(c, ac) \ p(ac, abc) = \frac{1}{4} \frac{5}{6} = \frac{5}{24} \neq p(c, abc)$
		\end{itemize}
	\end{counterexample}
	
	The inequalities show none of these sets qualifies as category. 
	
	\subsection{Independence of the axioms} \label{SEC:independence of the axioms}
	
	The tables below presents choice data for $X =\{a, b, c, x\}$. 
	
	\begin{center}
		\begin{tabular}{c | c  c}
			&	$a$ 	&	 $b$	\\
			\hline
			$p$		& 	$\frac{2}{5}$	&	$\frac{3}{5}$
		\end{tabular}
		\quad
		\begin{tabular}{c | c  c}
			&	$a$ 	&	 $c$	\\
			\hline
			$p$		& 	$\frac{3}{4}$	&	$\frac{1}{4}$
		\end{tabular}
		\quad
		\begin{tabular}{c | c  c}
			&	$b$ 	&	 $c$	\\
			\hline
			$p$		& 	$\frac{1}{3}$	&	$\frac{2}{3}$
		\end{tabular}
		\quad
		\begin{tabular}{c | c c c}
			&	$a$ 	&	 $b$	&	$c$	\\
			\hline
			$p$		& 	$\frac{1}{5}$	&	$\frac{3}{5}$ & $\frac{1}{5}$
		\end{tabular}
	\end{center}
	\begin{center}
		\begin{tabular}{c | c c }
			&	$a$  &	$x$  \\
			\hline
			$p$		& 	$\frac{1}{3}$	&	$\frac{2}{3}$ 
		\end{tabular}
		\quad
		\begin{tabular}{c | c c }
			&	$b$  &	$x$  \\
			\hline
			$p$		& 	$\frac{7}{8}$	&	$\frac{3}{8}$ 
		\end{tabular}
		\quad 
		\begin{tabular}{c | c c }
			&	$c$  &	$x$  \\
			\hline
			$p$		& 	$\frac{7}{8}$	&	$\frac{3}{8}$ 
		\end{tabular}
		\quad
		\begin{tabular}{c | c  c	c}
			&	$a$ 	&	 $b$  &  $x$\\
			\hline
			$p$		& 	$\frac{1}{4}$	&	$\frac{3}{8}$ & $\frac{3}{8}$
		\end{tabular}
	\end{center}
	\begin{center}
		\begin{tabular}{c | c  c	c}
			&	$a$ 	&	 $c$	& $x$	\\
			\hline
			$p$		& 	$\frac{15}{32}$	&	$\frac{5}{32}$ & $\frac{3}{8}$
		\end{tabular}
		\quad
		\begin{tabular}{c | c	c	c}
			&	$b$ 	&	 $c$ & $x$	\\
			\hline
			$p$		& 	$\frac{5}{24}$	&	$\frac{10}{24}$ & $\frac{3}{8}$
		\end{tabular}
		\quad
		\begin{tabular}{c | c c c c}
			&	$a$ 	&	 $b$	&	$c$		& $x$	\\
			\hline
			$p$		& 	$\frac{1}{8}$	&	$\frac{3}{8}$ & $\frac{1}{8}$ & $\frac{3}{8}$
		\end{tabular}
	\end{center}
	
	According to this choice data, \cind{} holds for $G = \{a, b, c\}$.
	However, observe that $p(x, ax) = \frac{2}{3} \ \neq \ \frac{3}{8} = p(x, acx)$.
	This is sufficient to determine failure of \cneu.
	It follows that \cind{} $ \ \nRightarrow \ $ \cneu{}.\\
	
	Consider now the following choice data for $X = \{a, b, x\}$.
	
	\begin{center}
		\begin{tabular}{c | c  c}
			&	$a$ 	&	 $b$	\\
			\hline
			$p$		& 	$\frac{2}{5}$	&	$\frac{3}{5}$
		\end{tabular}
		\quad
		\begin{tabular}{c | c  c}
			&	$b$ 	&	 $x$	\\
			\hline
			$p$		& 	$\frac{1}{3}$	&	$\frac{2}{3}$
		\end{tabular}
		\quad
		\begin{tabular}{c | c c }
			&	$a$  &	$x$  \\
			\hline
			$p$		& 	$\frac{1}{3}$	&	$\frac{2}{3}$ 
		\end{tabular}
		\quad
		\begin{tabular}{c | c c c}
			&	$a$ 	&	 $b$	&	$x$	\\
			\hline
			$p$		& 	$\frac{1}{6}$	&	$\frac{1}{6}$ & $\frac{2}{3}$
		\end{tabular}
	\end{center}

	Notice that \cneu{} holds for set $Z = \{a, b\}$. 
	Nonetheless, $\frac{p(a, ab)}{p(b, ab)} \ \neq \ \frac{p(a, abx)}{p(b, abx)}$, so that \cind{} is violated. 
	This shows that \cneu{} $\ \nRightarrow \ $ \cind.
	
	It follows that the two properties are independent.
	
	\section{Proofs}
	
	\subsection{Uniqueness of the coarsest partition} 
	
	Denote $\mathcal{P}$ the set of all partitions for which $p$ on $X$ is non-degenerate \scc.
	
	\begin{lemma} \label{Lemma: no intersection}
		For any two distinct partitions $\{X_i\}_{i \in I}, \ \{X_j\}_{j \in J} \in \mathcal{P}$, $X_i \cap X_j \neq \varnothing$ implies $X_i \subseteq X_j$ or $X_j \subseteq X_i$.
	\end{lemma}
	
	\begin{proof}
		We proceed by contradiction.
		Without loss of generality, fix one $j$ such that $X_i \cap X_j \neq \varnothing$, and assume $X_i \setminus X_j \neq \varnothing$ and $X_j \setminus X_i \neq \varnothing$. 
		Then there is a menu $A = \{a, b, c\}$ where $a, b \in X_i$ and $b, c \in X_j$.
		Since $p$ is \scc{} with respect to $\{X_i\}_{i \in I}$, it must be $p(a, abc) = p(a, ab) p(ab, abc)$.
		Similarly, $p$ is \scc{} with respect to $\{X_j\}_{j \in J}$, $p(a, abc) = p(a, ab)$. 
		Thus, $p(a, ab) p(ab, abc) = p(a, ab)$ and $p(ab, abc) = 1$, which violates positivity of $p$ and yields a contradiction.
		This concludes the proof.
	\end{proof}
	
	We write $\{X_i\}_{i \in I} \succsim \{X_j\}_{j \in J}$ when $\{X_i\}_{i \in I}$ is coarser than $\{X_j\}_{j \in J}$, namely, if for all $j \in J$ there exists $i \in I$ such that $X_j \subseteq X_i$. 
	
	\begin{lemma} \label{Lemma: coarser}
		For any two $\{X_i\}_{i \in I}, \{X_j\}_{j \in J} \in \mathcal{P}$ there exist $\{X_k\}_{k \in K} \in \mathcal{P}$ such that $\{X_k\}_{k \in K} \succsim \{X_i\}_{i \in I}$ and $\{X_k\}_{k \in K} \succsim \{X_j\}_{j \in J}$.
	\end{lemma}

	\begin{proof}
		We exhibit such partition and show it belongs to $\mathcal{P}$ by showing that $p$ is non-degenerate \scc. 
		
		Let $\{X_k\}_{k \in K} =  \{X_i \in \{X_i\}_{i \in I} \colon \nexists X_j \in  \{X_j\}_{j \in J} \text{ such that \ } X_i \subsetneq X_j\} \cup \{X_j \in \{X_j\}_{j \in J} \colon \nexists X_i \in  \{X_i\}_{i \in I} \text{ s.t.\ } X_j \subsetneq X_i\}$. 
		Thus $\{X_k\}_{k \in K}$ collects the largest sets of each partition. 
		Notice that if $\{X_i\}_{i \in I} \succsim \{X_j\}_{j \in J}$, then $\{X_k\}_{k \in K} = \{X_i\}_{i \in I}$, and the lemma follows trivially.
		Also, $\{X_k\}_{k \in K} \neq \{X\}$, since $\{X\}$ does not belong to $\{X_i\}_{i \in I}$ nor $\{X_j\}_{j \in J}$.
		Similarly, $\{X_k\}_{k \in K}$ cannot only contain singletons.
		We first show that $\{X_k\}_{k \in K}$ is a partition. 
		To see that $X_k \cap X_{k^\prime} = \varnothing$ for any $k, k^\prime \in K$, recall that by construction either $X_k, X_{k^\prime} \in \{X_i\}_{i \in I}$, and $X_k \cap X_{k^\prime} = \varnothing$; or $X_k\in \{X_i\}_{i \in I}$ and $X_{k^\prime} \in \{X_j\}_{j \in J}$, and $X_k \cap X_{k^\prime} = \varnothing$ by Lemma \ref{Lemma: no intersection}. 
		By construction, $\bigcup_{k \in K} X_k \subseteq X$. 
		Also, each $x \in X$ belongs to some $X_i, i \in I$ and to some $X_j, j \in J$. 
		Thus $X_i \cap X_j \neq \varnothing$ and it must be $X_i \subseteq X_j$ (or $X_j \subseteq X_i$) by Lemma \ref{Lemma: no intersection}.
		By construction, $X_k = X_j$ (or $X_k = X_i$) for some $k$, and $x \in \bigcup_{k \in K} X_k$. 
		It follows that $\cup_{k \in K} X_k = X$.

		We now show that $p$ is non-degenerate \scc{} w.r.t.\ $\{X_k\}_{k \in K}$. 
		Note that for any $X_i \subseteq X_j$, it must be $\sigma_j(a, A_j) = \omega_I(i, \pi_I(A_j)) \sigma_i(a, A_J \cap X_i)$, and $\omega_J(j, \pi_J(A)) = \displaystyle \sum_{i \text{ s.t. } X_i \subseteq X_j} \omega(i, \pi(A))$. 
		Let $\iota \colon I \to K$	and $\gamma \colon J \to K$ be maps that associate each set $X_i, X_j$ to the set $X_k$ that includes each. 
		For each $k \in K$, define $\sigma_k = \sigma_i$, if $X_k = X_i$ for some $i \in I$; or $\sigma_k = \sigma_j$ if $X_k = X_j$ for some $j$. 
		Similarly, for any $K^\prime \subseteq K$, let $\omega_K(k, K^\prime) = \omega_I(\iota^{-1}(k), \iota^{-1}(K^\prime))$	if $X_k = X_i$ for some $i \in I$, or $\omega_K(k, K^\prime) = \omega_J(\gamma^{-1}(k), \gamma^{-1}(K^\prime))$ if $X_k = X_j$ for some $j$. 
		As a result, for all $A \in \X$, for all $a \in A \in X_k$, $p(a, A) = \omega_K(k, \pi_K(A)) \ \sigma_K(a, A \cap X_k)$. 
	\end{proof}

	\uniquenessofcoarsestpartition*
	
	\begin{proof}	\label{PROOF:uniquenessofcoarsestpartition}
		It follows easily from Lemma \ref{Lemma: no intersection} and Lemma \ref{Lemma: coarser}.
	\end{proof}

		\subsection{Classification structures} \label{SEC: categorical structures}

	\categoriesintersection*

	\begin{proof}
		Towards a contradiction, assume $C \cap D \neq \varnothing$ and there are items $a, b, c \in X$ such that $a \in C \setminus D$, $b \in C \cap D$, $c \in D \setminus C$. 
		Since $D$ is a category, by \cneu\ it must be $p(a, ab) = p(a, abc)$, and equivalently $p(b, ab) = p(b, abc) + p(c, abc)$. 
		This entails that $p(b, ab) > p(b, abc)$, by positivity of $p$. 
		In turn, this implies that 
		$
		 \frac{p(a, ab)}{p(b, ab)} < \frac{p(a, abc)}{p(b, abc)},
		$
		which is a violation of \cind{} for $C$ and contradicts the hypothesis $C$ is a category. 
	\end{proof}
	
	\weakcategoriesintersection*
	
	\begin{proof}
		If $|W \cap T| = 1$, then $C \cap D$ trivially is a weak category. 
		If $|W \cap T| > 1$, consider any $S \subseteq W \cap T$ and $a, b \in S$. 
		Notice that, for any $E \subseteq (X \setminus W \cap T)$, $E = (E \cap T) \cup (E \cap (X \setminus T))$. 
		Since $a, b \in W$, \cind{} is satisfied and $\frac{p(a, S)}{p(b, S)} = \frac{p(a, S \cup (E \cap T))}{p(b, S \cup (E \cap T))}$. 
		Similarly, since $a, b \in T$, \cind{} is satisfied and $\frac{p(a, S \cup (E \cap T))}{p(b, S \cup (E \cap T))} = \frac{p(a, S \cup (E \cap T) \cup (E \cap (X \setminus T))}{p(b, S \cup (E \cap T) \cup (E \cap (X \setminus T)))}$. 
	\end{proof}


	\subsection{A characterization of \scc}

	\resolvability*
	
	\begin{proof} \label{PROOF:categorization}
		$ $ \newline
		[$\Leftarrow$].
		Let $p \colon X \times \X \to [0, 1]$ be a stochastic choice with categorization. 
		We point out that any $X_i$ s.t.\ $|X_i| \geqslant 2$ is a non-trivial category. 
		
		In fact, by Definition~\ref{DEF: resolvable functions}, $\forall A \in \X$, $\forall a, b \in A \cap X_i$, 
		$$
		\displaystyle \frac{p(a, A)}{p(b, A)} = \frac{\sigma_i(a, A \cap X_i) \ \omega(i, \pi(A))}{\sigma_i(b, A \cap X_i) \ \omega(i, \pi(A))} = \frac{\sigma_i(a, A \cap X_i)}{\sigma_i(b, A \cap X_i)} = \frac{p(a, A \cap X_i)}{p(b, A \cap X_i)}
		$$ 
		and \cind{} is satisfied. 
		
		Also, notice that for any $E \subseteq X \setminus X_i$, any $S \subseteq X_i$, $x \in E \cup X_j$ for some $j \neq i$,
		$$
		p(x, S \cup E) \ = \ \omega(j, \pi(S \cup E)) \ \sigma_j(x, E \cup X_j) = p(x, X_i \cup E),
		$$
		so that \cneu{} holds, as well.

		\noindent [$\Rightarrow$].	
		We show that the existence of a non-trivial category $C \subseteq X$ for a positive stochastic choice function $p \colon X \times \X \to [0,1]$ implies that $p$ is \scc.

		First, we build a partition of $X$. 
		By hypothesis, there exists a non-trivial category $C$ in $X$. 
		Let $C$ be the first subset $X_1$ in the partition. 
		Now look for a non-trivial category in $X \setminus X_1$. 
		If any, pick one and let it be the second subset $X_2$ in the partition, and proceed in the same way until no more non-trivial categories are found. 
		Assign remaining items to singletons, and collect indices in $I = \{1, \dots, k\}$.

		We now show that $p$ is \scc{} with respect to partition $\{X_i\}_{i \in I}$ built as above.
		For each $i \in I$, define $\sigma_i$ as the restriction of $p$ to $X_i$.
		That is, for all $B \subseteq X_i$, $a \in B$, $\sigma_i(a, B) = p(a, B)$.
		
		Consider any menu $A \in \X$, and denote $S = A \cap X_i$, $E = A \setminus X_i$ the set of items in $A$ within and outside category $X_i$, respectively. 	
		By Definition~\ref{DEF:general category}, \cind{} holds for $X_i$, so that the ratio $\frac{p(a, S \cup E)}{p(a, S)} = \Delta$ is fixed for all $a \in S$ to some $\Delta \in \mathbb{R}$. 
		It follows that
		$$
		\forall a \in S, \ \ p(a, S \cup E) = \sigma_i(a, S) \, \cdot \, \Delta.
		$$
		We define $\omega(i, \pi(S \cup E))  =  \Delta  =  p(S, S \cup E) \ = \ \displaystyle \sum_{a \in S} p(a, A)$.
		We show that $\omega$ is well-defined.
		In particular, we need to check that $\omega(i, \pi(A)) = \omega(i, \pi(B))$ for any $B \in \X$ s.t.\ $\pi(A) = \pi(B)$.
		Without loss of generality, fix $B$ and build $A^\prime = E \cup (B \cap X_i)$. 
		By \cneu, $p(x, A) = p(x, (A \cap X_i) \cup E) = p(x, (B \cap X_i) \cup E) = p(x, A^\prime)$ for all $x \in E$
		Thus, $p(A \cap X_i, A) = 1 - \sum_{x \in E} p(x, A) =  1 - \sum_{x \in E} p(x, A^\prime) = p(B \cap X_i, A^\prime)$. 
		Substituting $B \cap X_j$ to $A \cap X_j$ for each $j \in \pi(A)$ to obtain set $B$, and applying the same reasoning, we finally have $p(A \cap X_i, A) = p(B \cap X_i, B)$ for all $i \in \pi(A)$, and function $\omega$ is well-defined.	
	\end{proof}

	\subsection{A characterization of \scwc}
	
	\generalresolvability*
	
	\begin{proof} \label{PROOF: weak categorization}
		$ $ \newline
		[$\Leftarrow$].
		Let $p \colon X \times \X \to [0, 1]$ be a non-degenerate stochastic choice by weak categorization. 
		By Definition~\ref{DEF:weakcategorization}, $\forall A \in \X$, $\forall a, b \in A \cap X_i$, 
		$$
		\displaystyle \frac{p(a, A)}{p(b, A)} = \frac{\sigma_i(a, A \cap X_i) \ \omega_A(i, \pi(A))}{\sigma_i(b, A \cap X_i) \ \omega_A(i, \pi(A))} = \frac{\sigma_i(a, A \cap X_i)}{\sigma_i(b, A \cap X_i)} = \frac{p(a, A \cap X_i)}{p(b, A \cap X_i)},
		$$ 
		\cind{} is satisfied and $X_i$ is a non-trivial category.
		
		\noindent [$\Rightarrow$].	
		We show that the existence of a weak category $G \subseteq X$ for a stochastic choice function $p \colon X \times \X \to [0,1]$ implies that $p$ is \scwc.
		
		First, we build a partition of $X$. 
		By hypothesis, there exists a non-trivial weak category $G$ in $X$. 
		Let $G$ be the first subset $X_1$ in the partition. 
		Now look for a non-trivial category in $X \setminus X_1$. 
		If any, pick one and let it be the second subset $X_2$ in the partition, and proceed in the same way until no more non-trivial categories are found. 
		Assign remaining items to singletons, and collect indices in $I = \{1, \dots, k\}$.

		We now show that $p$ is \scwc\ with respect to partition $\{X_i\}_{i \in I}$ built above.	
		For each $i \in I$, define $\sigma_i$ as the restriction of $p$ to $X_i$.
		That is, for all $B \subseteq X_i$, $a \in B$, $\sigma_i(a, B) = p(a, B)$.
		
		Consider a menu $A \in \X$, and denote $S = A \cup X_i$, $E = A \setminus X_i$ the set of items in $A$ within and outside the weak category $X_i$, respectively. 	
		By Definition~\ref{DEF:general category}, \cind{} holds for $X_i$, so that the ratio $\frac{p(a, S \cup E)}{p(a, S)} = \Delta$ is fixed for all $a \in S$ to some $\Delta \in \mathbb{R}$. 
		It follows that
		$$
		p(a, S \cup E) = \sigma_i(a, S) \, \cdot \, \Delta.
		$$
		
		Define $\omega_{A}(i, \pi(S \cup E)) \ = \ p(S, S \cup E) \ = \ \displaystyle \sum_{a \in S} p(a, A)$, and notice that
		$$ 
		\displaystyle p(S, S \cup E) = \sum_{a \in S} p(a, S \cup E) = \sum_{a \in S} \sigma_i(a, S) \, \cdot \, \Delta = \Delta, 
		$$ 
		and as a result $p(a, A) = \sigma_i(a, A \cap X_i) \  \omega_{A}(i, \pi(A))$.
	\end{proof}

	\subsection{A generating population} \label{SEC:a population}
	
	\resolvabilitypopulation*
	
	\begin{proof} 
		$ $ \newline
		To prove this, it suffices to exhibit a probability distribution $Q$ on $\mathcal{R}$ that represents the \scc{} $p$.
		By definition, each deterministic choice function $c_X \in \mathcal{R}$ can be broken down into a choice $c_I$ on the set $I$ of indices, and a collection $\{c_i\}_{i \in I}$ of choices defined on sets in partition $\{X_i\}_{i \in I}$. 
		Denote $\I = 2^I \setminus \{\varnothing\}$ and $\X_i = 2^{X_i} \setminus \{\varnothing\}$ for each $i \in I$.
		As no other constraint is required, these choices can take any value and are mutually independent.
		
		Let $Q$ be defined as follows
		$$
		Q(c_X) \ = \ Q(c_I, \{c_i\}_{i \in I}) \ = \ \displaystyle \prod_{J \in \I} \omega(c_I(J), J) \prod_{i \in I} \prod_{S \in \X_i} \sigma_i(c_i(S), S).
		$$
		In other words, we model each choice as a set of independent events, and attach to each choice function the compound probability that all the defining events realize.
		To check that $Q$ is a probability distribution, it suffices to notice that for each choice $c_X$ there is a choice $c_X^\prime$ which differs only for one factor in the product. 
		Such factor corresponds to a choice in some menu, either in $\mathcal{I}$ or in some $\mathcal{X}_i$. 
		It follows that $\displaystyle \sum_{c_X \in \mathcal{R}}Q(c_X) = 1$, as can be checked by sequentially summing across all choices which differ only for one factor in a menu. 
		
		Notice now that $a \in X_j$ for some $j \in I$, so that
		\begin{align*}
			\sum_{c_X : c_X(A) = a} Q(c_X) \ & = \ \sum_{c_X : c_I(\pi(A)) = j \wedge c_i(A \cap X_j) = a} Q(c_I, \{c_i\}_{i \in I}) \\
			& = \ \sum_{c_I(\pi(A)) = j \wedge c_i(A \cap X_j) = a} \displaystyle \bigg( \prod_{J \in \I} \omega(c_I(J), J) \prod_{i \in I} \prod_{S \in \X_i} \sigma_i(c_i(S), S) \bigg)
		\end{align*}
		
		with $c_I(\pi(A)) = j$ and $c_i(A \cap X_j) = a$ for all products in the summation. 
		
		Thus we have 
		\begin{align*}
			& \sum_{c_X : c_X(A) = a} Q(c_X)\\
			& = \ \omega(j, \pi(A)) \sigma_j(a, A \cap X_j) \underbrace{\sum_{c_I(\pi(A)) = j \wedge c_i(A \cap X_j) = a} \displaystyle \Bigg( \prod_{J \subseteq \Omega_i \setminus \{j\}} \omega(c_I(J), J) \prod_{i \in I \setminus \{j\}} \prod_{S \subseteq X_i} \sigma_i(c_i(S), S)\Bigg)}_{= 1} \\
			& = \ p(a, A).
		\end{align*}
	\end{proof}

	\bigskip
	
	\population*
	
	\begin{proof}
		$ $ \newline
		$(2) \Rightarrow (1)$. Choices in $\mathcal{R}$ are such that 
		$$
		c_X(A) = c_{i}(A \cap X_i), \quad \text{with } \ i = c_I(\pi(A)),
		$$ 
		where $\pi \colon X \to I$ is the map that assigns each item to a set in partition $\{X_i\}_{i \in I}$.
		Let $q$ be a probability distribution on $\mathcal{R}$.
		Distribution $q$ induces a stochastic choice on $X$: $p(a, A) = q(\{c_X \colon c_X(A) = a\})$.
		By definition, 
		\begin{align*}
			p(A \cap X_i, A) = \sum_{a \in A \cap X_i} p(a, A) & = \sum_{a \in A \cap X_i} q(\{c_X \colon c_X(A) = a\})  \\
			& = q(\{c_X \colon c_X(A) \in A \cap X_i\})  \\
			& = q(\{c_X \colon c_I(\pi(A)) = i\}).
		\end{align*}
		
		Similarly, 
		\begin{align*}
			p(B \cap X_i, B) = \sum_{b \in B \cap X_i} p(b, B) & = \sum_{b \in B \cap X_i} q(\{c_X \colon c_X(B) = b\})  \\
			& = q(\{c_X \colon c_X(B) \in B \cap X_i\})  \\
			& = q(\{c_X \colon c_I(\pi(B)) = i\}).
		\end{align*}
		
		Since $\pi(A) = \pi(B)$ and $\{c_X \colon c_I(\pi(A)) = i\} = \{c_X \colon c_I(\pi(B)) = i\}$, it follows that the equality $p(A \cap X_i, A) \  =  \ p(B \cap X_i, B)$ holds.
		
		\bigskip
		
		$(1) \Rightarrow (2)$. 
		To prove this implication we exhibit a probability distribution $Q$ on the set $\mathcal{R}$ of deterministic choice functions resolvable with respect to $\{X_i\}_{i \in I}$. 
		It is convenient to first define a stochastic choice on the set of indices, $\omega \colon I \times \I \to [0, 1]$, as follows:
		$$
		\omega(i, J) \ = \ p(A \cap X_j, A), \text{ for some } A \text{ such that } \pi(A) = J.
		$$
		By $(1)$, $\omega$ is well-defined. 
		Now, write $p(a, A | A \cap X_i)$ the probability that $a \in X_i$ is chosen from $A$, conditional on the event that an item in $A \cap X_i$ is chosen. 
		Namely, $p(a, A | A \cap X_i) \ = \ \frac{p(a, A)}{p(A \cap X_i, A)} \ = \ \frac{p(a, A)}{\omega(i, \pi(J))}$. 
		Finally, define $Q$ as follows
		$$
		Q(c_X) \ = \ \displaystyle \prod_{B \in \X} p(c_X(B), B) \ = \ \prod_{B \in \X} p(c_X(B), B | B \cap X_{i = \pi(c_X(B))}) \omega(\pi(c_X(B)), \pi(B)).
		$$
		\noindent We therefore have 
		\begin{align*}
			p(a, A) \ & = \ \sum_{c_X \colon c_X(A) = a} Q(c_X)  \\
			& = \sum_{c_X \colon c_X(A) = a} \Bigg(\prod_{B \in \X} p(c_X(B), B | B \cap X_{i = \pi(c_X(B))}) \omega(\pi(c_X(B)), \pi(B)) \Bigg) \\
			& =  p(c_X(A), A | A \cap X_i) \omega(i, \pi(A)),
		\end{align*}
		where the last equality follows from considering that
		$$
		\sum_{c_X \colon c_X(A) = a} \Bigg( \prod_{B \in \X \setminus A} p(c_X(B), B | B \cap X_{i = \pi(c_X(B))}) \  \omega(\pi(c_X(B)), \pi(B))\Bigg) \ = \ 1.
		$$
	\end{proof}

	\subsection{Stochastic choices with categorization and \rum}

	\RUMstability*
	
	\begin{proof} \label{PROOF:categorizationRUM}
		$ $	\newline
		\noindent [$\Rightarrow$].
		Since $\omega \colon I \times \I \to [0,1]$ is \rum, there exists a probability distribution $\mathit{v}$ defined on the set $\mathcal{L}_I$ of linear orders on $I$ that provides a random utility representation $(\mathit{v}, \mathcal{L}_I)$ of $\omega$.
		Since each $\sigma_i$ is \rum{}, each has a \rum{} representation $(\mathit{s}_i, \mathcal{L}_i)$, where $\mathcal{L}_i$ is the set of linear orders on $X_i$, and $\mathit{s}_i$ is a probability distribution on $\mathcal{L}_i$.
		
		Now, let $\mathcal{L}_X$ be the set of all linear orders on $X$, and $\overline{\mathcal{L}_X}$ its subset of orders $L_X$ such that 
		$$
		(\exists a, b \in X \ \text{s.t.} \ a \in X_i, \ b \in X_j, \ \ a L_X b) \ \ \Rightarrow \ \ x L_X y \ \ \forall x \in X_i, \forall y \in X_j.
		$$
		Linear orders in $\overline{\mathcal{L}_X}$ unequivocally rank categories. 
		With some abuse of notation, we write $\pi(L_X)$ the ranking of categories induced by $L_X \in \overline{\mathcal{L}_X}$. 
		Also, we write $L_{X \upharpoonright i}$ to denote the binary relation induced by $L_X$ on items in $X_i$.
		We claim that the couple $(\mathit{q}, \mathcal{L}_X)$, with $\mathit{q}$ defined as follows:
		$$
		\mathit{q}(L_X) = 
		\begin{cases}
			\mathit{v}(\pi(L_X)) \ \prod_{i \in I} \mathit{s}_i(L_{X \upharpoonright i}), & \text{if } L_X \in \overline{\mathcal{L}_X}, \\
			0, & \text{otherwise},
		\end{cases}
		$$
		is a \rum{} representation of $p$. 
		It suffices to observe that 
		\begin{align*}
			p(a, A) \ & =  \ \omega(i, \pi(A)) \ \sigma_i(a, A \cap X_i) \\
			& \displaystyle = \ \bigg( \sum_{L_I \in \mathcal{L}_I : i = \max(\pi(A), L_I)}\mathit{v}(L_I)\bigg) \ \bigg(\sum_{L_i : a = \max(A \cap X_i, L_i)}\mathit{s}_i(L_i) \bigg)
			 =  \displaystyle \sum_{L_X : a = \max(A, L_X)} \mathit{q}(L_X). 
		\end{align*}
		where the last equality is given by the following procedure. 
		First, fix some $L_J$ such that $i = \max(\pi(A), L_J)$ and some $L_j$ with $a = \max(A \cap X_i, L_j)$ and consider all linear orders $L_X$ such that $L_J = \pi(L_X)$ and $L_j = L_{X \upharpoonright j}$. 
		Clearly, $a = \max(A, L_X)$ for any $L_X$ of this kind. 
		\begin{align*}
			\sum_{\substack{L_X : L_J = \pi(L_X) \\ L_{X \upharpoonright j} = L_j}} \displaystyle \mathit{q}(L_X)  \ & = \ \sum_{\substack{L_X : L_J = \pi(L_X) \\ L_{X \upharpoonright j} = L_j}}
			 v(\pi(L_X)) \prod_{i \in I} s_i(L_{X \upharpoonright i})\\[0.5em]
			& = \ \displaystyle \sum_{\substack{L_X : L_J = \pi(L_X) \\ L_{X \upharpoonright j} = L_j}}
			 v(L_J) s(L_j) \prod_{i \in I \setminus j} s_i(L_{X \upharpoonright i}) \\[0.5em] 
			& = \  v(L_J) s(L_j) \sum_{\substack{L_X : L_J = \pi(L_X) \\ L_{X \upharpoonright j} = L_j}} \Bigg(\prod_{i \in I \setminus j} s_i(L_{X \upharpoonright i})\Bigg).
		\end{align*}
		Let us focus on the last summation. 
		Each addend is a product of probabilities, each one attached to a linear order defined on a set $X_i$ in the partition. 
		In particular, we sum the products given by all possible combinations of such linear orders.
		Divide the sum as follows. 
		Consider first combinations in which all linear orders in the first $|I - 2|$ indices are the same, and differ for the $|I - 1| - th$ one. 
		Summing the products of combinations that share the first $|I - 2|$ factors, the common part in the product can be factored out, while the remaining factors sum up to $1$. 
		In this way, we obtain products that are shorter by one factor. 
		To sum them, proceed in the same fashion and sum first those which share the same first $|I - 3|$ factors. 
		By the same reasoning, shorter products will be obtained. 
		Applying this method repeatedly, we obtain that 
		$$
		\sum_{\substack{L_X : L_J = \pi(L_X) \\ L_{X \upharpoonright j} = L_j}} \Bigg(\prod_{i \in I \setminus j} s(L_{X \upharpoonright i}) \Bigg) = 1.
		$$
		It follows that 
		$$
		\displaystyle \sum_{L_X : a = \max(A, L_X)} \mathit{q}(L_X) \ = \sum_{\substack{L_I : i = \max(\pi(A), L_I) \\ L_i : a = \max(A \cap X_i, L_i)}} v(L_I) s(L_i) \ = \bigg( \sum_{L_I : i = \max(\pi(A), L_I)}\mathit{v}(L_I)\bigg) \ \bigg(\sum_{L_i : a = \max(A \cap X_i, L_i)}\mathit{s}_i(L_i) \bigg).
		$$
		Apply the same method to show that $q$ is a probability distribution, i.\ e., $\displaystyle \sum_{L_X \in \mathcal{L}_X}q(L_X) = 1$.
		
		\bigskip
		
		\noindent [$\Leftarrow$].
		Since $p$ is \scc{} and \rum, $\sigma_i$ is \rum{} as well, being a restriction of $p$ to $X_i$, for all $i \in I$.
		
		By definition, $\omega(i, \pi(A)) = \displaystyle \sum_{a \in A \cap X_i} p(a, A)$, for all $A \in \X, \ i \in I$.
		
		Let $\mathcal{L}$ be the set of linear orders on $X$. 
		Since $p$ is \rum, there exists a probability distribution $\mathit{q}$ on $\mathcal{L}$ that represents $p$.
		
		Also, since $p$ is \scc,
		\begin{align*}
			\omega(i, \pi(A)) \ & =  \ \displaystyle \sum_{a \in A \cap X_i} p(a, A) \ = \ \mathit{q}(\{L : \max(A, L) \in A \cap X_i\}) \\ 
			& = \displaystyle \sum_{L : \max(A, L) \in A \cap X_i} \mathit{q}(L) \ = \ \omega(i, \pi(B))
		\end{align*} 
		for all $B$ s.t.\ $\pi(B) = \pi(A)$.
		
		Without loss of generality, consider a menu $E \in \X$ such that $|E| = 3$, and $E = \{a, b, c\}$ with $a, b \in X_i$ for some $i \in I$ and $c \in X_j$ for some $j \neq i$. 
		Then, 6 orderings are possible. 
		
		\begin{center}
%
		\begin{tabular}{ c | c | c | c | c | c | c }
			order & $a L  b  L  c$ & $a  L  c  L  b$ & $b L a L c$ & $b L c L a$ & $c L a L b$ & $c L b L a$ \\
			\hline
			probability & $\alpha$ & $\beta$ & $\gamma$ & $\delta$ & $\epsilon$ & $\mu$
		\end{tabular}
		\end{center}
		
		Since $ \pi(E) \ = \ \pi(ac) = \ \pi(bc)$, by definition of \scc{} we have
		\begin{align*}
			\omega(i, \pi(E)) \ = \ \omega(i, \pi(ac)) \quad \Leftrightarrow & \quad \alpha + \beta + \gamma + \delta =  \alpha + \beta + \gamma
			\quad 	\Rightarrow \delta = 0 \\
			\text{ and } \quad 	\omega(i, \pi(ac)) \ = \ \omega(i, \pi(bc)) \quad \Leftrightarrow & \quad \alpha + \beta + \gamma  =  \alpha + \gamma + \delta
			\quad \Rightarrow \beta = 0.
		\end{align*}\
		
		Define $\overline{\mathcal{L}} \subseteq \mathcal{L}$ the set of linear orders on $X$ s.t.\ 
		$$
		(a \in X_i, \ b \in X_j, \ \ a L b) \ \ \Rightarrow \ \ x L y \ \ \forall x \in X_i, \forall y \in X_j,
		$$
		and notice that $\beta$ and $\delta$ are probabilities associated to linear orders $L \notin \overline{\mathcal{L}}$.
		Since both must be null, we conclude $\nexists L \in \mathcal{L} \setminus \overline{\mathcal{L}}$ s.t.\ $q(L) > 0$.
		It follows that $\omega$ has a \rum{} representation $(v, \mathcal{L}_I)$, where $\mathcal{L}_I$ is the set of linear orders $L_I$ on $I$ and $\mathit{v}(L_I) = q(\{L : \max(X, L) \in \pi^{-1}(\max(I, L_I))\})$.	 
	\end{proof}
	
	\bigskip
	
	\localrationality* 
	
	\begin{proof} \label{local rationalizability}
		$ $	\newline
		In order for $G$ to be a non-trivial weak category, \cind{} must be satisfied. 
		Since $p$ on $X$ is locally rationalizable at $G \subseteq X$, there exists a probability distribution $q$ on the set of linear orders $\mathcal{L}_G$ on $G$ such that, for all $x \in G$, $A \subseteq X$, 
		$$
		p(x, A) \ = \ p(G, A) \ \ q(\{\ell \in \mathcal{L}_G \colon \max(A \cap G, \ell) = x\}).
		$$  
		Consider any menu $A$, and let $S = G \cap A$ and $E = A \ \setminus G$. 
		For any $a, b \in S$, 
		$$
		\displaystyle \frac{p(a, S \cup E)}{p(b, S \cup E)} \ = \ \frac{p(S, A) \ \ q(\{\ell \in \mathcal{L}_G \colon \max(S, \ell) = x\})}{p(S, A) \ \ q(\{\ell \in \mathcal{L}_G \colon \max(S, \ell) = x\})} \ = \frac{p(a, S)}{p(b, S)},
		$$ 
		so that \cind{} holds and $G$ is a non-trivial weak category for $p$. 
	\end{proof}

\end{appendices}


\bibliographystyle{apalike}
\bibliography{referencesssss}  

@article{ApesteguiaBallester2013,
	author  = {Apesteguia, J. and Ballester, M. A.},
	title   = {Choice by Sequential Procedures},
	journal = {Games and Economic Behavior},
	year    = {2013},
	volume  = {77},
	number  = {1},
	pages   = {90--99}
}

@article{Aguiar2017,
	author  = {Aguiar, V. H.},
	title   = {Random categorization and bounded rationality},
	journal = {Economics Letters},
	year    = {2017},
	volume  = {159},
	pages   = {46--52}
}

@article{Barbos2010,
	author  = {Barbos, A.},
	title   = {Context Effects: A Representation of Choices from Categories},
	journal = {Journal of Economic Theory},
	year    = {2010},
	volume  = {145},
	number  = {3},
	pages   = {1224--1243}
}

@incollection{BlockMarshack1960,
	author    = {Block, H. D. and Marschak, J.},
	title     = {Random orderings and stochastic theories of responses},
	booktitle = {Contributions to probability and statistics},
	year      = {1960},
	pages     = {97--132},
	publisher = {Stanford Univ. Press}
}

@article{Cantoneetal.2020,
	author  = {Cantone, D. and Giarlotta, A. and Watson, S.},
	title   = {Choice resolutions},
	journal = {Social Choice and Welfare},
	year    = {2020},
	volume  = {56},
	number  = {4},
	pages   = {713--753},
	url     = {https://link.springer.com/article/10.1007/s00355-020-01285-9}
}

@article{CarpentiereDoignon2025,
	author  = {Carpentiere, D. and Doignon, J. P.},
	title   = {Resolutions of probabilistic choice spaces},
	journal = {Theory and Decision},
	year    = {2025},
	volume = {99},
	pages = {225--253},
	doi     = {10.1007/s11238-025-10064-w}
}

@article{Chernevetal2015,
	author  = {Chernev, A. and Böckenholt, U. and Goodman, J.},
	title   = {Choice overload: a conceptual review and meta-analysis},
	journal = {Journal of Consumer Psychology},
	year    = {2015},
	volume  = {25},
	number  = {2},
	pages   = {333--358}
}

@article{Debreu1960,
	author  = {Debreu, G.},
	title   = {Review of R.D. Luce, Individual Choice Behavior: a Theoretical Analysis},
	journal = {American Economic Review},
	year    = {1960},
	volume  = {50},
	number  = {1},
	pages   = {186--188}
}

@article{Echeniqueetal2019,
	author  = {Echenique, F. and Saito, K.},
	title   = {General Luce model},
	journal = {Economic Theory},
	year    = {2019},
	volume  = {68},
	number  = {4},
	pages   = {811--826}
}

@article{EcheniqueSaitoTserenjigmid2018,
	author  = {Echenique, F. and Saito, K. and Tserenjigmid, G.},
	title   = {The perception-adjusted Luce model},
	journal = {Mathematical Social Sciences},
	year    = {2018},
	volume  = {93},
	pages   = {67--76}
}

@article{Falmagne1978,
	author  = {Falmagne, J. C.},
	title   = {A representation theorem for finite random scale systems},
	journal = {Journal of Mathematical Psychology},
	year    = {1978},
	volume  = {18},
	pages   = {52--72}
}

@article{Faro2023,
	author  = {Faro, J. H.},
	title   = {Luce Model with Replicas},
	journal = {Journal of Economic Theory},
	year    = {2023},
	volume  = {208},
	pages   = {1055--1096}
}

@article{Fricketal2019,
	author  = {Frick, M. and Iijima, R. and Strzalecki, T.},
	title   = {Dynamic Random Utility},
	journal = {Econometrica},
	year    = {2019},
	volume  = {87},
	number  = {6},
	pages   = {1941--2002}
}

@article{GulPesendorfer2001,
	author  = {Gul, F. and Pesendorfer, W.},
	title   = {Temptation and Self-Control},
	journal = {Econometrica},
	year    = {2001},
	volume  = {69},
	number  = {6},
	pages   = {1403--1435}
}

@article{GulNatenzonPesendorfer2014,
	author  = {Gul, F. and Natenzon, P. and Pesendorfer, W.},
	title   = {Random choice as behavioral optimization},
	journal = {Econometrica},
	year    = {2014},
	volume  = {82},
	pages   = {1873--1912}
}

@article{Heydari2020,
	author  = {Heydari, P.},
	title   = {Stochastic Choice over Menus},
	journal = {Theory and Decision},
	year    = {2020},
	volume  = {88},
	number  = {2},
	pages   = {257--268}
}

@article{HuberPaynePuto1982,
	author  = {Huber, J. and Payne, J. W. and Puto, C.},
	title   = {Adding asymmetrically dominated alternatives: violations of regularity and similarity hypothesis},
	journal = {The Journal of Consumer Research},
	year    = {1982},
	volume  = {9},
	number  = {1},
	pages   = {90--98}
}

@article{IyengarLepper2000,
	author  = {Iyengar, S. S. and Lepper, M. R.},
	title   = {When choice is demotivating: can one desire too much of a good thing?},
	journal = {Journal of Personality and Social Psychology},
	year    = {2000},
	volume  = {79},
	pages   = {995--1006}
}

@article{Kovach2022,
	author  = {Kovach, M. and Tserenjigmid, G.},
	title   = {Behavioral Foundations of Nested Stochastic Choice and Nested Logit},
	journal = {Journal of Political Economy},
	year    = {2022},
	volume = {130}, 
	pages = {2411--2461}
}

@article{Kreps1979,
	author  = {Kreps, D. M.},
	title   = {A Representation Theorem for Preference for Flexibility},
	journal = {Econometrica},
	year    = {1979},
	volume  = {47},
	number  = {3},
	pages   = {565--577}
}

@article{KrepsPorteus1979,
	author  = {Kreps, D. M. and Porteus, E. L.},
	title   = {Dynamic choice theory and dynamic programming},
	journal = {Econometrica},
	year    = {1978},
	volume  = {47},
	number  = {1},
	pages   = {91--100}
}

@article{LlerasMasatliogluNakajima2017,
	author  = {Lleras, J. S. and Masatlioglu, Y. and Nakajima, D. and Ozbay, E. Y.},
	title   = {When More Is Less: Limited Consideration},
	journal = {Journal of Economic Theory},
	year    = {2017},
	volume  = {170},
	pages   = {70--85}
}

@misc{LiTang2016,
	author       = {Li, J. and Tang, R.},
	title        = {Associationistic Luce Rule},
	year         = {2016},
	howpublished = {Working Paper}
}

@book{Luce1959,
	author    = {Luce, R. D.},
	title     = {Individual Choice Behavior: A Theoretical Analysis},
	publisher = {Wiley},
	address   = {New York},
	year      = {1959}
}

@article{ManziniMariotti2007,
	author  = {Manzini, P. and Mariotti, M.},
	title   = {Sequentially rationalizable choice},
	journal = {American Economic Review},
	year    = {2007},
	volume  = {97},
	number  = {5},
	pages   = {1824--1839}
}

@article{ManziniMariotti2012,
	author  = {Manzini, P. and Mariotti, M.},
	title   = {Categorize then choose: boundedly rational choice and welfare},
	journal = {Journal of the European Economic Association},
	year    = {2012},
	volume  = {10},
	number  = {5},
	pages   = {1141--1165}
}

@article{Masatlioglu2016,
	author  = {Masatlioglu, Y. and Nakajima, D. and Ozbay, E.},
	title   = {Revealed attention},
	journal = {American Economic Review},
	year    = {2016},
	volume  = {102},
	number  = {5},
	pages   = {2183--2205}
}

@misc{McFadden1978,
	author    = {McFadden, D.},
	title     = {Modeling the choice of residential location},
	year      = {1978},
	howpublished = {Transportation Research Record}
}

@article{Simonson1989,
	author  = {Simonson, I.},
	title   = {Choice based on reasons: the case of attraction and compromise effects},
	journal = {Journal of Consumer Research},
	year    = {1989},
	volume  = {16},
	number  = {2},
	pages   = {158--174}
}

@article{Thurstone1927,
	author  = {Thurstone, L. L.},
	title   = {A law of comparative judgement},
	journal = {Psychological Review},
	year    = {1927},
	volume  = {34},
	number  = {4},
	pages   = {273--286}
}

@article{Tversky1972,
	author  = {Tversky, A.},
	title   = {Elimination by aspects: a theory of choice},
	journal = {Psychological Review},
	year    = {1972},
	volume  = {79},
	number  = {4},
	pages   = {281--299}
}

@article{Tyson2013,
	author  = {Tyson, C. J.},
	title   = {Behavioral implications of shortlisting procedures},
	journal = {Social Choice and Welfare},
	year    = {2013},
	volume  = {41},
	number  = {4},
	pages   = {941--963}
}

@misc{Zhang2016,
	author       = {Zhang, J.},
	title        = {Stochastic choice with subjective categorization},
	year         = {2016},
	howpublished = {Available at SSRN}
}

\end{document}